\documentclass[envcountsame,a4paper,12pt]{llncs}
\usepackage[margin=2.3cm,marginpar=1.6cm]{geometry}
\def\Vhrulefill{\leavevmode\leaders\hrule height 0.7ex depth \dimexpr0.4pt-0.7ex\hfill\kern0pt}

\usepackage{lineno}

\usepackage{makeidx} 

\usepackage{microtype,scalefnt,hyperref}

\usepackage{relsize}%

\usepackage[xspace]{ellipsis}
\usepackage{listings}
\usepackage{tikz}
\usepackage{proof}
\usepackage{enumerate}
\usepackage{graphicx}
\usepackage{scalerel}
\usepackage{booktabs}
\usepackage{array}
\usepackage{arydshln}
\usepackage{xcolor}
\usepackage{stmaryrd}

\usepackage{comment}

\usepackage{amsmath,amssymb}

\usepackage{thmtools,thm-restate}
\usepackage{centernot}
\usepackage{mathtools}
\usepackage{stmaryrd}

\newcommand*\patchAmsMathEnvironmentForLineno[1]{%
	\expandafter\let\csname old#1\expandafter\endcsname\csname #1\endcsname
	\expandafter\let\csname oldend#1\expandafter\endcsname\csname end#1\endcsname
	\renewenvironment{#1}%
	{\linenomath\csname old#1\endcsname}%
	{\csname oldend#1\endcsname\endlinenomath}}%
\newcommand*\patchBothAmsMathEnvironmentsForLineno[1]{%
	\patchAmsMathEnvironmentForLineno{#1}%
	\patchAmsMathEnvironmentForLineno{#1*}}%
\AtBeginDocument{%
	\patchBothAmsMathEnvironmentsForLineno{equation}%
	\patchBothAmsMathEnvironmentsForLineno{align}%
	\patchBothAmsMathEnvironmentsForLineno{flalign}%
	\patchBothAmsMathEnvironmentsForLineno{alignat}%
	\patchBothAmsMathEnvironmentsForLineno{gather}%
	\patchBothAmsMathEnvironmentsForLineno{multline}%
}

\newcommand{\qedd}{\qed}

\title{Polyteam Semantics
\thanks{This research was supported by the Marsden grant UOA1628, administered by the Royal Society of New Zealand, and the grants 292767 and 308712 of the Academy of Finland. The third author was an international research fellow of Japan Society for the Promotion of Science (Postdoctoral
Fellowships for Research in Japan (Standard)).}}

\author{Miika Hannula\inst{1,2} \and Juha Kontinen\inst{2} \and Jonni~Virtema\inst{2,3,4}}

\institute{
University of Auckland, New Zealand
\and
University of Helsinki, Finland,
\email{\{miika.hannula,juha.kontinen\}@helsinki.fi}
\and
Hasselt University, Belgium
\and
Hokkaido University, Japan,
\email{jonni.virtema@let.hokudai.ac.jp}
}

\newcommand{\complClFont}[1]{\mathsf{#1}}         
\newcommand{\logicClFont}[1]{\mathsf{#1}}        
\newcommand{\mathCommandFont}[1]{\mathrm{#1}}     

\newcommand{\ar}[1]{{\protect\ensuremath{\mathCommandFont{ar}(#1)}}}
\newcommand{\rel}[1]{{\protect\ensuremath{\mathCommandFont{rel}(#1)}}}
\newcommand{\tuple}[1]{\vec{#1}}
\newcommand{\sub}{\subseteq}
\newcommand{\Dom}[1]{{\protect\ensuremath{\mathsf{Dom}(#1)}}}
\newcommand {\indep}[3] {#2 ~\bot_{#1}~ #3}
\newcommand {\indepc}[2] {#1 ~\bot~ #2}

\providecommand{\dfn}{\mathrel{\mathop:}=}
\providecommand{\ddfn}{\mathrel{\mathop{{\mathop:}{\mathop:}}}=}

\newcommand{\PTIME}{\protect\ensuremath{\complClFont{PTIME}}\xspace}

\newcommand{\calC}{\protect\ensuremath{\mathcal{C}}}

\newcommand{\calL}{\protect\ensuremath{\mathcal{L}}}

\newcommand{\calR}{\protect\ensuremath{\mathcal{R}}}

\newcommand{\ESO}{\logicClFont{ESO}}

\newcommand{\FO}{\logicClFont{FO}}
\newcommand{\PFO}{\logicClFont{PFO}}

\newcommand{\posgfp}{\logicClFont{PosGFP}}

\newcommand{\dep}[1]{=\!\left(#1\right)}
\newcommand{\pdep}[4]{=\!\left(#1,#2/#3,#4\right)}
\newcommand{\pcon}[2]{=\!\left(#1/#2\right)}
\newcommand{\pinc}[2]{#1 \sub #2}
\newcommand{\exc}[2]{#1 \mid #2}
\newcommand{\deps}{\rm dep}
\newcommand{\pdeps}{\rm pdep}
\newcommand{\incs}{\rm inc}
\newcommand{\pincs}{\rm pinc}
\newcommand{\pinds}{\rm pind}
\newcommand{\inds}{\rm ind}
\newcommand{\excs}{\rm exc}
\newcommand{\pexcs}{\rm pexc}

\newcommand{\ci}[3]{#2~\bot_{#1}~#3}
\newcommand{\cixyz}{\ci{\tuple x}{\tuple y}{\tuple z}}

\newcommand{\cvees}{\hspace{-.4mm}\cvee\hspace{-.4mm}}
\newcommand{\cveee}{\hspace{.5mm}\varovee\hspace{.5mm}}

\newcommand{\cvee}{\mbox{\larger[1.5]$\cveee$}}

\renewcommand{\vec}{\overline}

\newcommand{\var}[1]{\textrm{Var}(#1)}
\newcommand{\fr}[1]{\textrm{Fr}(#1)}
\newcommand{\ifr}[2]{\textrm{Fr}_{#1}(#2)}
\newcommand{\emp}[2]{{#1}_{#2\mapsto \emptyset}}
\newcommand{\rr}{\mathCommandFont{r}}

\renewcommand{\vec}{\overline}
\newcommand{\pt}[1]{\vec{#1}}

\newcommand{\A}{\mathfrak{A}}
\newcommand{\Po}{\mathcal{P}}
\newcommand{\N}{\mathbb{N}}

\newenvironment{redtext}{\color{red}}{\ignorespacesafterend}
\newenvironment{bluetext}{\color{blue}}{\ignorespacesafterend}
\newenvironment{purpletext}{\color{purple}}{\ignorespacesafterend}

\newcommand\ScaleExists[1]{\vcenter{\hbox{\scalefont{#1}$\exists$}}}

\DeclareMathOperator*\bigexists{%
	\vphantom\sum
	\mathchoice{\ScaleExists{2}}{\ScaleExists{1.4}}{\ScaleExists{1}}{\ScaleExists{0.75}}}

\newcommand\ScaleForall[1]{\vcenter{\hbox{\scalefont{#1}$\forall$}}}

\DeclareMathOperator*\bigforall{%
	\vphantom\sum
	\mathchoice{\ScaleForall{2}}{\ScaleForall{1.4}}{\ScaleForall{1}}{\ScaleForall{0.75}}}

\begin{document}

\maketitle

\begin{abstract}
Team semantics is the mathematical framework of modern logics of dependence and independence in which formulae are interpreted by sets of assignments (teams) instead of single assignments as in first-order logic. In order to deepen the fruitful interplay  between team semantics and database dependency theory, we define  \emph{Polyteam Semantics} in which formulae are evaluated over a family of teams. We begin by defining a novel polyteam variant of dependence atoms and  give a finite  axiomatisation for the associated implication problem.
We relate polyteam semantics to team semantics and investigate in which cases logics over the former can be simulated by logics over the latter.
We also  characterise the expressive power of poly-dependence logic by properties of polyteams that are downwards closed and definable in existential second-order logic ($\ESO$).  The analogous result is shown to hold for poly-independence logic and all $\ESO$-definable properties.
We also relate poly-inclusion logic to greatest fixed point logic.
\end{abstract}

\section{Introduction}
Team semantics is the mathematical framework of modern logics of dependence and independence.  The origin of team semantics goes back to  \cite{hodges97} but its development to its current form began with the publication  of the monograph \cite{vaananen07}. In team semantics formulae are interpreted by sets of assignments (teams) instead of single assignments as in first-order logic. The reason for this change is that statements such as  \emph{the value of a variable $x$ depends on the value of  $y$}  do not really make sense for single assignments. Team semantics has interesting  connections with database theory and  database dependencies \cite{DBLP:conf/lpar/Hannula15,DBLP:conf/foiks/HannulaK14,DBLP:conf/cikm/HannulaKL14,KontinenLV13}. In order to facilitate the exchange  between team semantics and   database theory, we introduce a generalisation of team semantics  in which formulae are evaluated over a family of teams. 
We identify a natural notion of poly-dependence that generalises dependence atoms to polyteams and give a finite axiomatisation for its implication problem. We also define polyteam versions of independence, inclusion and exclusion atoms, and characterise the expressive power of logics using these novel atoms.

A team $X$  is a set of assignments with a common finite domain  $x_1,\ldots, x_{n}$ of variables. Such a team  can be  viewed as a database table with $x_1,\ldots, x_{n}$ as its attributes. Dependence logic extends the language of first-order logic with atomic formulae $\dep{\tuple x, y}$ called \emph{dependence atoms} 
 expressing that the value of the variable $y$ is functionally determined by the values of the variables in $\tuple x$. On the other hand, \emph{independence atoms} 
$\tuple y \perp_{\tuple x} \tuple 
z$
\cite{gradel13}  
express that, for any fixed  value of $ \tuple x$, knowing the value of $\tuple z$ does not tell us anything new about the value of $\tuple y$. By viewing a team as a database table, the atoms  $\dep{\tuple x, y}$ and  $\tuple y \perp_{\tuple x} \tuple 
z$ correspond  to the widely studied  functional and embedded multivalued dependencies. Furthermore,   inclusion atoms  $\tuple x \subseteq \tuple y$  and exclusion atoms $ \tuple x | \tuple y$  of \cite{galliani12} inherit their semantics from the corresponding  database dependencies. 

Independence, inclusion, and exclusion atoms  have very interesting properties in the team semantics setting. For example, inclusion atoms give rise to a variant of dependence logic that corresponds to the complexity class PTIME over finite ordered structures \cite{gallhella13} whereas all the other atoms above (and their combinations) give rise to  logics that are equi-expressive with existential second-order logic and the complexity class NP. The complexity theoretic aspects of logics in team semantics have been studied extensively during the past few years  (see \cite{DKV} for a survey).  

A multiset version of team semantics was recently defined in \cite{2015arXiv151009040D}. Multiteam semantics is motivated by the fact that multisets are widely assumed in database theory and occur in applications. 
Multiteam semantics widens the applicability of team semantics for the study of qualitative (e.g., functional dependence) and quantitative (e.g., probabilistic independence) dependencies in a unified framework. Recently multiteam semantics was further generalised to so-called probabilistic team semantics in which information on exact multiplicities are replaced by a (probability) distribution over the assignments \cite{HKMV18,jelia19}.

The aim of this work is similar to that of  \cite{2015arXiv151009040D}, i.e.,  we want to  extend the applicability of team semantics.  In database theory dependencies are often expressed by so-called embedded dependencies. An \emph{embedded dependency} is a sentence of first-order logic with equality of the form
\[
\forall x_1\dots\forall x_n \big(\phi(x_1,\dots,x_n) \rightarrow \exists y_1\dots\exists y_k \psi(x_1,\dots,x_n,y_1,\dots,y_k)\big),
\]
where $\phi$ and $\psi$ are conjunctions of relational atoms $R(x_1,\dots,x_n)$ and equalities $x=y$. In the literature embedded dependencies have been thoroughly classified stemming from real life applications. Examples of well-known subclasses include \emph{full}, \emph{uni-relational}, \emph{1-head}, \emph{tuple-generating}, and \emph{equality-generating}. For example, an embedded dependency is called \emph{tuple-generating} if it is equality free (for further details see, e.g., \cite[Section 3]{Kanellakis90}). The uni-relational dependencies can be studied also in the context of team semantics as generalised dependencies \cite{kuusisto12}. However in many applications, especially in the area of data exchange and data integration, it is essential to be able to express dependencies between different relations.
 
 In the context of data exchange (see, e.g., \cite{fagin05}) the relational database is divided into a set of source relations $\mathcal S$ and a set of target relations  $\mathcal T$. Dependencies are used to describe what kind of properties should hold when data is transferred from the source schema to the target schema. In this setting a new taxonomy of embedded dependencies rises: An embedded dependency  $\forall \tuple x\big ( \phi(\tuple x) \rightarrow \exists \tuple y \psi(\tuple x, \tuple y) \big )$ is \emph{source-to-target} if the relation symbols occurring in $\phi$ and $\psi$ are from  $\mathcal S$ and  $\mathcal T$, respectively. The embedded dependency is \emph{target} if the relation symbols occurring in it are from  $\mathcal T$. There is no direct way to study these classes of dependencies in the uni-relational setting of team semantics. 
 In this paper we propose a general framework in which these inherently poly-relational dependencies can be studied.

 In Section \ref{sec:unitopoly} we  lay the foundations of polyteam semantics.  The shift to polyteams is exemplified in Section \ref{sec:polydep}  by the definition of poly-dependence atoms and an Armstrong type axiomatisation for the associated implication problem. In Section \ref{sec:pteamsemantics} polyteam semantics is extended from atoms to complex formulae.  Section \ref{sect:rel}  studies the relationship between polyteam and team semantics. We show that some polyteam logics can be simulated  with team logics. Section \ref{sec:expressivity} examines the expressive power of the new logics over polyteams. We characterise  poly-independence (poly-dependence)  logic as the maximal logic capable of  defining  all (downwards closed) properties of polyteams definable in existential second-order logic. Finally we relate poly-inclusion logic to a fragment of greatest fixed point logic.

\section{From uni-dependencies to poly-dependencies}\label{sec:unitopoly}
We start by defining the familiar dependency notions from the team semantics literature. In Section $\ref{sec:polydep}$ we introduce a novel poly-relational version of dependence atoms and establish a finite axiomatisation of its implication problem. We then continue to present poly-relational versions of inclusion, exclusion, and independence atoms, and a general notion of a poly-relational dependency atom. We conclude this section by relating the embedded dependencies studied in database theory to our new setting.

\subsection{Dependencies in team semantics}\label{sec:teams}
Vocabularies $\tau$ are  sets of relation symbols with prescribed arities. For each $R\in\tau$, let $\ar{R}\in Z_+$ denote the arity of $R$. A $\tau$-structure is a tuple $\A = \big(A, (R^\A_i)_{R_i\in\tau}\big)$, where $A$ is a set and each $R^\A_i$ is an $\ar{R_i}$-ary relation on $A$ (i.e., $R^\A_i \subseteq A^{\ar{R_i}}$). We use $\A$, $\mathfrak{B}$, etc. to denote $\tau$-structures and $A$, $B$, etc.\ to denote the corresponding domains.

Let $D$ be a finite set of first-order variables and $A$ be a nonempty set. A function $s\colon D \to A$ is called an \emph{assignment}. For a variable $x$ and $a\in A$, the assignment $s(a/x)\colon D\cup\{x\} \rightarrow A$ is obtained from $s$ as follows:
\[
s(a/x)(y) :=
\begin{cases}
a & \text{if $y=x$},\\
s(y) &\text{otherwise}. 
\end{cases} 
\]

A $\emph{team}$ $X$ is a set of assignments with a common domain $D$ and codomain $A$. If $\A$ is a $\tau$-structure and $X$ a team with codomain $A$, then we say that $X$ is a team of $\A$. Let $\tuple x=(x_1,\ldots,x_n)$ be a sequence of variables, and $s$ an assignment. We write $s(\tuple x)$ to denote the sequence $\big(s(x_1),\ldots ,s(x_n)\big)$, and $X(\tuple x)$ for the set of values $\{s(\tuple x)\mid s\in X\}$. For a set (or sequence) of variables $V$, we write $X\upharpoonright V$ for the restriction of $X$ to $V$.

The following dependency atoms were introduced in \cite{vaananen07,galliani12,gradel13}.

\begin{definition}[Dependency atoms]
Let $\A$ be a structure and $X$ a team with codomain $A$. If $\tuple x, \tuple y$ are variable sequences, then $\dep{\tuple x,\tuple y}$ is a dependence atom with the truth condition:
$$\A \models_X \dep{\tuple x,\tuple y} \textrm{ if for all } s,s'\in X \textrm{ s.t.\ }s(\tuple x) = s'(\tuple x), \textrm{ it holds that } s(\tuple y)=s'(\tuple y).$$

If $\tuple x,\tuple y$ are variable sequences of the same length, then $\tuple x \subseteq \tuple y$ is an inclusion atom and $\exc{\tuple{x}}{\tuple{y}}$ an exclusion atom with satisfaction defined as follows:
\begin{align*}
&\A \models_X \tuple x \sub \tuple y \textrm{ if for all }s\in X \textrm{ there exists } s'\in X\textrm{ such that } s(\tuple x)=s'(\tuple y).\\
&\A \models_X \exc{\tuple x}{\tuple y} \textrm{ if for all }s,s'\in X: s(\tuple x)\neq s'(\tuple y).
\end{align*}
If $\tuple x,\tuple y,\tuple z$ are variable sequences, then $\cixyz$ is a conditional independence atom with satisfaction defined by
\begin{align*}&\A\models_{X} \cixyz \textrm{ if for all } s,s'\in X\textrm{ such that }s(\tuple x)=s'(\tuple x)\textrm{ there exists } s''\in X \\&\textrm{ such that }s''(\tuple x)=s(\tuple x)\textrm{, }s''(\tuple y)=s(\tuple y), \textrm{ and }s''(\tuple z)=s'( \tuple z). 
\end{align*}
\end{definition}

Note that in the previous definitions it is allowed that some or all of the vectors of variables have length $0$. For example, $\A\models_X\dep{\tuple y}$ holds iff $\forall s\in X: s(\tuple y)=\tuple c$ holds for some fixed tuple $\tuple c$, where $\dep{\tuple y}$ denotes the dependence atom $\dep{\tuple x,\tuple y}$ such that $\tuple x$ is of length $0$. Also, $\A\models_X\cixyz$ holds always if either of the vectors $\tuple y$ or $\tuple z$ is of length $0$.

All the aforementioned dependency atoms have corresponding variants in relational databases. One effect of this relationship is that the axiomatic properties of these dependency atoms trace back  to well-known results in database theory. Armstrong's axioms for functional dependencies constitute a finite axiomatisation for dependence atoms \cite{armstrong74,gradel13}, and inclusion atoms can be finitely axiomatised using the axiomatisation for inclusion dependencies  \cite{DBLP:journals/jcss/CasanovaFP84}. On the other hand, the non-axiomatisability and undecidability of the (finite and unrestricted) implication problem for embedded multivalued dependencies both carry over to conditional independence atoms \cite{herrmann95,Sagiv:1982,parker:1980}.
Restricting attention to the so-called \emph{pure independence atoms}, i.e., atoms of the form $\indep{\emptyset}{\tuple x}{\tuple y}$, a finite axiomatisation is obtained by relating to marginal independence in statistics \cite{geiger:1991,KontinenLV13}.

 \subsection{The notion of poly-dependence}\label{sec:polydep}
For each $i\in\mathbb{N}$, let $\var{i}$ denote a distinct countable set of first-order variable symbols. We say that these variables are of sort $i$. Relating to databases, sorts correspond to table names. Usually we set $\var{i}=\{x^i_j\mid j\in\mathbb{N}\}$. We write $x^i$, $y^i$,  $x^i_j$ to denote variables from $\var{i}$, and $\tuple x^i$ to denote tuples of variables from $\var{i}$. Sometimes we drop the index $i$ and write simply $x$ and  $\tuple x$ instead of $x^i$ and  $\tuple x^i$, respectively. Note that $\tuple x$ is always a tuple of variables of a single sort. In order to simplify  notation, we sometimes  write $\tuple{x}^i$ and $\tuple{x}^j$ to denote arbitrary tuples of variables of sort $i$ and $j$, respectively. We emphasise that $\tuple{x}^i$ and $\tuple{x}^j$ might be of different length and may consist of distinct variables. Let $\A$ be  a $\tau$-structure and let $D_i\subseteq\var{i}$ for all $i\in \N$. A tuple $\vec{X}=(X_i)_{i\in\mathbb{N}}$ is a \emph{polyteam} of $\A$ with domain $\tuple D =(D_i)_{i\in \N}$, if  $X_i$ is a team with domain $D_i$ and co-domain $A$ for each $i\in\mathbb{N}$. We identify $\pt X$ with $(X_0, \ldots ,X_n)$ if $X_i$ is the singleton team consisting of the empty assignment for all $i$ greater than $n$.
Let $\vec{X}=(X_i)_{i\in\mathbb{N}}$ and $\vec{Y}=(Y_i)_{i\in\mathbb{N}}$ be two polyteams. We say that $\tuple X$ is a \emph{subteam} of $\vec{Y}$ if $X_i\sub Y_i$ for all $i\in \mathbb{N}$. By the \emph{union} (resp. \emph{intersection}) of $\tuple X$ and $\tuple Y$ we denote the polyteam $(X_i\cup Y_i)_{i\in\mathbb{N}}$ (resp. $(X_i\cap Y_i)_{i\in\mathbb{N}}$). By a slight abuse of notation we write $\tuple X \cup \tuple Y$ (resp. $\tuple X \cap \tuple Y$) for the union (resp. intersection) of $\tuple X$ and $\tuple Y$, and $\tuple X \sub \tuple Y$ to denote that $\tuple X$ is a subteam of $\tuple Y$. For a tuple $\tuple V=(V_i)_{i\in\mathbb{N}}$ where $V_i\sub \var{i}$, the \emph{restriction} of $\tuple X$ to $\tuple V$, written $\tuple X \upharpoonright \tuple V$, is defined as $(X_i\upharpoonright V_i)_{i\in\mathbb{N}}$.

Next we  generalise dependence atoms to the polyteam setting. In contrast to the standard dependence atoms, poly-dependence atoms declare functional dependence of variables over two teams.

 \noindent
\textbf{Poly-dependence}.
Let $\tuple x^i\tuple y^i$ and $\tuple u^j \tuple v^j$ be sequences of variables such that  $\tuple x^i$ and $\tuple u^j$, and $\tuple y^i$ and $\tuple v^j$ have the same length, respectively.
Then $\pdep{\tuple x^i}{\tuple y^i}{\tuple u^j}{\tuple v^j}$ is a \emph{poly-dependence atom} whose satisfaction relation $\models_{\pt X}$ is defined as follows:
\[\A \models_{\pt X} \pdep{\tuple x^i}{\tuple y^i}{\tuple u^j}{\tuple v^j}\Leftrightarrow\forall s\in X_i\forall s'\in X_j: s(\tuple x^i)=s'(\tuple u^j)\text{ implies }s( \tuple y^i)=s'( \tuple v^j) .\]
Note that the atom $\pdep{\tuple x}{\tuple y}{\tuple x}{\tuple y}$ corresponds  to the  dependence atom $\dep{\tuple x, \tuple y}$. For empty tuples $\tuple x^i$ and $\tuple u^j$ the  poly-dependence atom  reduces to a``poly-constancy atom'' $\pcon{\tuple y^i}{\tuple v^j}$.
We will later show (Remark \ref{uniremark}) that poly-dependence atoms of the form $\pdep{\tuple x^i}{\tuple y^i}{\tuple u^i}{\tuple v^i}$ can be expressed with formulae using only ordinary dependence atoms.  Thus poly-dependence atoms of this form are considered as primitive notions only when $\tuple x^i\tuple y^i = \tuple u^i \tuple v^i$; otherwise $\pdep{\tuple x^i}{\tuple y^i}{\tuple u^i}{\tuple v^i}$ is considered as a shorthand for the equivalent formula obtained from Remark \ref{uniremark}.

The ability to reason about database dependencies facilitates many data management tasks such as schema design, query optimisation, and integrity maintenance. Keys, inclusion dependencies, and functional dependencies in particular have a crucial role in all of these processes. A traditional way to approach the interaction between dependencies has been the utilisation of proof systems similar to natural deduction systems in logic.  The most significant of all these systems is  Armstrong's axiomatisation for functional dependencies. This inference system consists of only  three rules which we depict below using the standard notation for functional dependencies, i.e., $X\to Y$ denotes that an attribute set $X$ functionally determines another attribute set $Y$. 
\begin{definition}[Armstrong's axiomatisation \cite{armstrong74}]\label{armstrong}
\begin{itemize}
\item Reflexivity: If $Y\sub X$, then $X\to Y$
\item Augmentation: if $X \to Y$, then $XZ\to YZ$
\item Transitivity: if $X \to Y$ and $Y\to Z$, then $X\to Z$
\end{itemize}
\end{definition}

Our first objective is to generalise Armstrong's axiomatisation to the poly-dependence setting. To this end, we assemble the three rules of Armstrong and introduce three auxiliary rules: Union, Symmetry, and Weak Transitivity. Contrary to the Armstrong's proof system, here Union is not reducible to Transitivity and Augmentation because we operate with sequences instead of sets of variables or attributes. Symmetry in turn is imposed by the sequential notation employed by the poly-dependence atom. Weak Transitivity exhibits transitivity of equalities on the right-hand side of a poly-dependence atom, a phenomenon that arises only in the polyteam setting.

\begin{definition}[Axiomatisation for poly-dependence atoms]\label{polyaxioms}
\begin{itemize}
\item Reflexivity: $ \pdep{\tuple x^i}{\text{pr}_k (\tuple x^i)}{\tuple y^j}{\text{pr}_k(\tuple y^j)}$, where $k=1, \ldots ,|\tuple x^i|$ and $\text{pr}_k$ takes the $k$th projection of a sequence.
\item Augmentation: if $\pdep{\tuple x^i}{\tuple y^i}{\tuple u^j}{\tuple v^j}$, then $ \pdep{\tuple x^i\tuple z^i}{\tuple y^i\tuple z^i}{\tuple u^j\tuple w^j}{\tuple v^j\tuple w^j}$
\item Transitivity: if $\pdep{\tuple x^i}{\tuple y^i}{\tuple u^j}{\tuple v^j}$ and $\pdep{\tuple y^i}{\tuple z^i}{\tuple v^j}{\tuple w^j}$, then $\pdep{\tuple x^i}{\tuple z^i}{\tuple u^j}{\tuple w^j}$
\item Union: if $\pdep{\tuple x^i}{\tuple y^i}{\tuple u^j}{\tuple v^j}$ and $\pdep{\tuple x^i}{\tuple z^i}{\tuple u^j}{\tuple w^j}$ then $\pdep{\tuple x^i}{\tuple y^i\tuple z^i}{\tuple u^j}{\tuple v^j\tuple w^j}$
\item Symmetry: if $\pdep{\tuple x^i}{\tuple y^i}{\tuple u^j}{\tuple v^j}$, then  $\pdep{\tuple u^j}{\tuple v^j}{\tuple x^i}{\tuple y^i}$
\item Weak Transitivity: if $\pdep{\tuple x^i}{\tuple y^i\tuple z^i\tuple z^i}{\tuple u^j}{\tuple v^j\tuple v^j\tuple w^j}$, then  $\pdep{\tuple x^i}{\tuple y^i}{\tuple u^j}{\tuple w^j}$
\end{itemize}
\end{definition}

This proof system forms a complete characterisation of logical implication for poly-dependence atoms. 
We use $\models$ to refer to logical implication, i.e., we write $\Sigma\models \sigma$ if $\A \models_{\pt X} \Sigma $ implies $ \A \models_{\pt X} \sigma$ for all  models $\A$ and polyteams $\pt X$. Given an \emph{axiomatisation} $\calR$, that is, a set of axioms and inference rules, we write $\Sigma \vdash_{\calR}\sigma$ if $\calR$ yields a proof of $\sigma$ from $\Sigma$. Given a class of dependency atoms $\calC$, we then say that $\calR$ is sound (complete, resp.) for $\calC$ if for all finite sets of dependency atoms $\Sigma\cup\{\sigma\}$ from $\calC$, $\Sigma\vdash_{\calR}\sigma $ implies (is implied by, resp.) $ \Sigma\models \sigma$.

\begin{theorem}
The axiomatisation of Def. \ref{polyaxioms} is sound and complete for poly-dependence atoms.
\end{theorem}
\begin{proof}
The proof of soundness is straightforward and omitted. We show that the axiomatisation is complete, i.e., that $\Sigma\models \sigma$ implies $\Sigma\vdash \sigma $ for a set $\Sigma\cup\{\sigma\}$ of poly-dependence atoms. Assume $\sigma$ is $\pdep{\tuple x^i}{\tuple y^i}{\tuple x^j}{\tuple y^j}$. First we consider the case where $i=j$ in which case $\sigma$ is a standard dependence atom. Let $\Sigma^*$ be the subset of $\Sigma$ consisting of all standard dependence atoms over $\var{i}$. Since all teams satisfying $\Sigma^*$ can be extended to a polyteam satisfying $\Sigma$ by introducing new empty teams, we have that $\Sigma^*\models \sigma$ in the team semantics setting. Since dependence atoms $\dep{\tuple x,\tuple y}$ in team semantics correspond to functional dependencies $\{x\in \tuple x^i\}\rightarrow \{y\in \tuple y^i\}$ in relational databases (see e.g. \cite{gradel13}), Armstrong's complete axiomatisation from Definition \ref{armstrong} yields a deduction of $\sigma_0$ from $\Sigma^*_0$ where $\Sigma^*_0$ and $\{\sigma_0\}$ are obtained from $\Sigma^*$ and $\sigma$ by replacing dependence atoms with their corresponding functional dependencies. Since dependence atoms are provably order-independent (i.e. one derives $\dep{\tuple x_0,\tuple x_1}$ from $\dep{\tuple y_0,\tuple y_1}$ by Reflexivity, Union, and Transitivity if $\tuple x_i$ and $\tuple y_i$ list the same variables), the deduction in Armstrong's system can be  simulated with the rules in Definition \ref{polyaxioms}. This proves the case $i=j$.

Let us then consider the case $i\neq j$. We will show that  $\Sigma\not\vdash \sigma $ implies $\Sigma\not\models \sigma$. Assume $\Sigma\not\vdash \sigma$.
Define first a binary relation $\sim$ on $\var{i}\cup\var{j}$ such that $a^i\sim a^j$ if $\Sigma\vdash \pdep{\tuple x^i}{a^i}{\tuple x^j}{a^j}$,  $a^j\sim a^i$ if $\Sigma\vdash \pdep{\tuple x^j}{a^j}{\tuple x^i}{a^i}$, and $a^i\sim b^i$ ($a^j\sim b^j$, resp.) if $a^i=b^i$ or $\Sigma\vdash \pdep{\tuple x^i}{a^ib^i}{\tuple x^j}{a^ja^j}$ for some $a^j$ ($a^j=b^j$ or $\Sigma\vdash \pdep{\tuple x^j}{a^jb^j}{\tuple x^i}{a^ia^i}$ for some $a^i$, resp.). We show that $\sim$ is an equivalence relation.
\begin{itemize}
\item Reflexivity: Holds by definition.
\item Symmetry: First note that $a^i\sim a^j$ and $a^j\sim a^i$ are derivably equivalent by the symmetry rule. Assume then that $a^i\sim b^i$ in which case $\pdep{\tuple x^i}{a^ib^i}{\tuple x^j}{a^ja^j}$ is derivable for some $a^j$. Then derive $\pdep{a^ib^i}{b^i}{a^ja^j}{a^j}$ and $\pdep{a^ib^i}{a^i}{a^ja^j}{a^j}$ by using the reflexivity rule, and then $\pdep{\tuple x^i}{b^i}{\tuple x^j}{a^j}$ and $\pdep{\tuple x^i}{a^i}{\tuple x^j}{a^j}$ by using the transitivity rule. Finally derive $\pdep{\tuple x^i}{b^ia^i}{\tuple x^j}{a^ja^j}$ by using the union rule.
\item Transitivity: Assume first that $a^i\sim b^i\sim c^i$, where $a^i, b^i,c^i$ and are pairwise distinct. Then  $\pdep{\tuple x^i}{a^ib^i}{\tuple x^j}{a^ja^j}$ and $\pdep{\tuple x^i}{b^ic^i}{\tuple x^j}{b^jb^j}$ are derivable for some $a^j$ and $b^j$. Then analogously to the previous case assemble $\pdep{\tuple x^i}{a^ib^ib^i}{\tuple x^j}{a^ja^jb^j}$ which admits $\pdep{\tuple x^i}{a^i}{\tuple x^j}{b^j}$ by weak transitivity, and detach $\pdep{\tuple x^i}{c^i}{\tuple x^j}{b^j}$ from $\pdep{\tuple x^i}{b^ic^i}{\tuple x^j}{b^jb^j}$.  By the union rule we then obtain $\pdep{\tuple x^i}{a^ic^i}{\tuple x^j}{b^jb^j}$ and thus that $a^i\sim c^i$. Since all the other cases are analogous, we observe that $\sim$ is transitive.
\end{itemize}

Let $s$ be a function that maps each $x\in \var{i}\cup\var{j}$ that appears in $\Sigma \cup\{\sigma\}$ to the equivalence class $x/\sim$. We define $\pt X=(X_i,X_j)$ where $X_k=\{s\upharpoonright \var{k}\}$ for $k=i,j$. First notice that $\pt X\not\models \sigma$ for, by union,  it cannot be the case that $\text{pr}_k(\tuple y^i)\sim \text{pr}_k(\tuple y^j)$ for all $k=1, \ldots ,|\tuple y^i|$. It suffices to show that $\pt X$ satisfies each $\pdep{\tuple u^m}{\tuple v^m}{\tuple u^n}{\tuple v^n}$ in $\Sigma$. If $m=n$ or $\{m,n\}\neq \{i,j\}$,  the atom is trivially satisfied. Hence, and by symmetry, we may assume that the atom is of the form $\pdep{\tuple u^i}{\tuple v^i}{\tuple u^j}{\tuple v^j}$. Assume that $s(\tuple u^i)=s(\tuple u^j)$, that is, $\text{pr}_k(\tuple u^i)\sim \text{pr}_k(\tuple u^j)$ for all $k=1, \ldots ,|\tuple u^i|$. We obtain by the union rule that $\pdep{\tuple x^i}{\tuple u^i}{\tuple x^j}{\tuple u^j}$ is derivable, and hence by the transitivity rule that  $ \pdep{\tuple x^i}{\tuple v^i}{\tuple x^j}{\tuple v^j}$ is also derivable. Therefore, by using the reflexivity and transitivity rules we conclude that $s(\tuple v^i)=s(\tuple v^j)$. \qedd
\end{proof}

 
 \subsection{A general notion of a poly-dependency}\label{sec:polygen}
Next we consider suitable polyteam generalisations for the dependencies discussed in Section \ref{sec:teams}  and also define a general notion of poly-dependency. This generalisation is immediate  for inclusion atoms which are  inherently multi-relational; relational database management systems maintain referential integrity by enforcing inclusion dependencies specifically between two distinct tables. With poly-inclusion atoms these multi-relational features can now be expressed. 

\noindent
\textbf{Poly-inclusion.}
Let $\tuple x^i$ and $\tuple y^j$ be sequences of variables of the same length. Then $\tuple x^i \sub \tuple y^j$ is a \emph{poly-inclusion atom} whose satisfaction relation $\models_{\pt X}$ is defined as follows:
\[\A \models_{\pt X} \pinc{\tuple x^i}{\tuple y^j} \Leftrightarrow \forall s\in X_i\exists s'\in X_j: s(\tuple x^i)=s'(\tuple y^j).\]
If $i=j$, then the atom is the standard inclusion atom.

\noindent
\textbf{Poly-exclusion.}
Let $\tuple x^i$ and $\tuple y^j$ be sequences of variables of the same length. Then $\exc{\tuple x^i}{ \tuple y^j}$ is a
\emph{poly-exclusion atom} whose satisfaction relation $\models_{\pt X}$ is defined as follows:
\[\A \models_{\pt X} \exc{\tuple x^i}{\tuple y^j} \Leftrightarrow \forall s\in X_i, s'\in X_j: s(\tuple x^i)\neq s'(\tuple y^j).\]
If $i=j$, then the atom is the standard exclusion atom.

\noindent 
\textbf{Poly-independence}
Let $\tuple x^i$, $\tuple y^i$, $\tuple a^j$,$\tuple b^j$, $\tuple u^k$, $\tuple v^k$, and $\tuple w^k$ be tuples of variables such that $|\tuple x^i|=|\tuple a^j|=|\tuple u^k|$, $|\tuple y^i|=|\tuple v^k|$, $|\tuple b^j|=|\tuple w^k|$. Then $\indep{\tuple{x}^i,\tuple{a}^j / \tuple u^k}{\tuple y^i / \tuple v^k}{\tuple b^j / \tuple w^k}$ is a
\emph{poly-independence atom} whose satisfaction relation $\models_{\pt X}$ is defined as follows:
\begin{align*}
&\A \models_{\pt X} \indep{\tuple{x}^i,\tuple{a}^j / \tuple u^k}{\tuple y^i / \tuple v^k}{\tuple b^j / \tuple w^k} \Leftrightarrow \forall s\in X_i, s'\in X_j: s(\tuple x^i)= s'(\tuple a^j)\text{ implies }\\
&\exists s''\in X_k: s''(\tuple u^k \tuple v^k)=s(\tuple x^i \tuple y^i) \text{ and } s''(\tuple w^k)=s'(\tuple b^j).
\end{align*}
The atom $\indep{\tuple{x},\tuple{x} / \tuple x}{\tuple y / \tuple y}{\tuple z / \tuple z}$, where all variables are of the same sort, corresponds to the standard independence atom $\indep{\tuple x}{\tuple y}{\tuple z}$. Furthermore, a \emph{pure poly-independence atom} is an atom of the form $\indep{\emptyset,\emptyset / \emptyset}{\tuple y^i / \tuple v^k}{\tuple b^j / \tuple w^k}$, written using a shorthand $\indep{}{\tuple y^i / \tuple v^k}{\tuple b^j / \tuple w^k}$.

Poly-independence atoms are closely related to equi-join operators of relational databases as the next example exemplifies. 
\begin{example}\label{ex:1}
A relational database schema
\begin{align*}
\textsc{P(rojects)}=&\{\texttt{project,team}\}, \quad  \textsc{T(eams)}=\{\texttt{team,employee}\},\\
\textsc{E(mployees)}=&\{\texttt{employee,team,project}\},
\end{align*}
stores information about distribution of employees for teams and projects in a workplace. The poly-independence atom
\begin{equation}\label{eqA}
\indep{\textsc{P}[\texttt{team}], \textsc{T}[\texttt{team}] / \textsc{E}[\texttt{team}]} {\textsc{P}[\texttt{project}] / \textsc{E}[\texttt{project}]} {\textsc{T}[\texttt{employee}] / \textsc{E}[\texttt{employee}]}
\end{equation}
expresses that the relation \textsc{Employees} includes as a subrelation the natural join of \textsc{Projects} and \textsc{Teams}. If furthermore $\textsc{E}[\texttt{project,team}]\sub \textsc{P}[\texttt{project,team}]$ and $ \textsc{E}[\texttt{team,employee}]\sub \textsc{T}[\texttt{team,employee}]$ hold, then \textsc{Employees} is exactly this natural join.
\end{example}

In addition to the poly-atoms described above, we define the notion of  generalised poly-atoms that is analogous to the notion of generalised  atoms of \cite{kuusisto12}.

\noindent
\textbf{Generalised poly-atoms.} Let $n\in \mathbb{N}$, and let $(j_1, \ldots ,j_n)$ be a sequence of positive integers. A \emph{generalised quantifier}
 of type $(j_1, \ldots ,j_n)$ is a collection $Q$ of relational structures $(A,R_1, \ldots ,R_n)$ (where each $R_i$ is $j_i$-ary) that is closed under isomorphisms. For every sequence $(\tuple x_1, \ldots ,\tuple x_n)$, where $\tuple x_i$ are length $j_i$ tuples  of variables from some $\var{l_i}$, $\alpha_Q(\tuple x_1,\ldots ,\tuple x_n)$ is a 
 \emph{generalised poly-atom}  of type $(j_1, \ldots ,j_n)$ and
 of sort $\{l_1, \ldots ,l_n\}$. For a structure $\A$ and polyteam $\pt X$ where $\tuple x_i\sub \Dom{X_{l_i}}$, the satisfaction relation with respect to $\alpha_Q(\tuple x_1,\ldots ,\tuple x_n)$ is defined as follows:
\begin{multline*}
\A\models_{\pt X} \alpha_Q(\tuple x_1,\ldots ,\tuple x_n) 
\Leftrightarrow \Big(\Dom{\A}, R_1:=\rel{X_{l_1}, \tuple x_1} \ldots ,R_n:=\rel{X_{l_n}, \tuple x_n}\Big)\in Q.
\end{multline*}
By $\rel{X, \tuple x}$, for $\tuple x = (x_1, \ldots ,x_m)$, we denote the relation $\{(s(x_1),\ldots ,s(x_m))\mid s\in X\} $. 
A generalised poly-atom $\alpha_Q(\tuple x_1,\ldots ,\tuple x_n)$ that has a singleton sort is called a \emph{uni-atom}.
When referring to the set of all poly-atoms of the form $\alpha_Q(\tuple x_1,\ldots ,\tuple x_n)$, for a fixed $Q$, we omit the tuples $\tuple x_1,\ldots ,\tuple x_n$ and write simply poly-atom $\alpha_Q$.
We say that an atom $\alpha_Q$ is definable in a logic $\calL$ if the class $Q$ is definable in $\calL$.
For instance, we notice that poly-inclusion atoms of the form $(x^i,y^i)\sub (u^j,v^j)$ are first-order definable generalised poly-atoms of type $(2,2)$.

\subsection{Database dependencies as poly-atoms}

Embedded dependencies in a multi-relational context can now be studied with the help of generalised poly-atoms and polyteam semantics. Conversely, strong results obtained in the study of database dependencies can be transferred and generalised for stronger results in the polyteam setting.
In particular, each embedded dependency can be seen as a defining formula for a generalised poly-atom, and hence the classification of embedded dependencies naturally yield a corresponding classification of generalised poly-atoms.
For example, the class
\begin{align*}
\mathcal{C}:=\{\alpha_Q(\tuple x_1,\ldots ,\tuple x_n) \mid \,&Q \text { is definable by an $\FO(R_1,\dots,R_n)$-sentence in}\\
&\text{the class of equality-generating dependencies}\}
\end{align*}
is the class of \emph{equality-generating} poly-atoms.  The defining formula of the generalised atom of type (2,2) that captures the poly-dependence atom of type $\pdep{x^i}{y^i}{u^j}{v^j}$ is
\[
\forall x_1 \forall x_2\forall y_1\forall y_2  \big((R_1(x_1,x_2) \land R_2(y_1,y_2) \land x_1=y_1) \rightarrow x_2=y_2\big).
\]
Thus poly-dependence atoms are included in the class of equality-generating poly-atoms. 

In order to study data exchange in the polyteam setting, we first need to define the notions of \emph{source-to-target} and \emph{target} poly-atoms. This classification of poly-atoms requires some more care as it is not enough to consider the defining formulae of the corresponding atoms, but also the variables that the atom is instantiated with. We will return to this topic briefly after we have given semantics for logics that work on polyteams.

\section{Polyteam semantics for complex formulae}\label{sec:pteamsemantics}
We  next delineate a version of team semantics suitable for the polyteam context.
We note here that it is not a priori clear what sort of modifications for connectives and quantifiers one should entertain when shifting from teams to the polyteam setting.

\subsection{Syntax and semantics}

\begin{definition}
Let $\tau$ be a set of relation symbols. 
The syntax of \emph{poly first-order logic} $\PFO(\tau)$ is given by the following grammar rules: 
\[
\phi \ddfn  x= y \mid  x \neq y \mid R({\vec{x}}) \mid  \neg R({\vec{x}})  \mid (\phi\land\phi) \mid (\phi\lor\phi)  \mid  (\phi\lor^j\phi) \mid \exists x \phi \mid  \forall x \phi,
\]
where $R\in\tau$ is a $k$-ary relation symbol, $j\in\mathbb{N}$, $\vec{x}\subseteq \var{i}^k$ and $x,y\in\var{i}$ for some $i,k\in\mathbb{N}$. 
\end{definition}
We say that $\vee$ is a \emph{global disjunction} whereas $\vee^i$ is a \emph{local disjunction}. A literal is said to be of sort $i$ if its variables are of sort $i$. Note that in the definition the scope of negation is restricted to atomic formulae.  Note also that the restriction of $\PFO(\tau)$ to formulae without the connective $\lor^j$ and using only variables of a single fixed sort is $\FO(\tau)$.

For the definition of polyteam semantics of $\PFO$, recall the definitions of teams and polyteams from Sections \ref{sec:teams} and \ref{sec:polydep}, respectively.
Let $X$ be a team, $A$ a non-empty set, and $F\colon X\to \Po(A)\setminus \{\emptyset\}$ a function. We denote by $X[A/x]$ the modified team $\{s(a/x) \mid s\in X, a\in A\}$, and by $X[F/x]$ the team $\{s(a/x)\mid s\in X, a\in F(s)\}$. Moreover let $\pt X$ be a polyteam. Then $\pt {X}[X/X_i]$ denotes the polyteam $(\dots,X_{i-1},X,X_{i+1},\dots)$.

Note that if restricted to the aforementioned single-sort fragment of $\PFO(\tau)$ the polyteam semantics below coincides with traditional team semantics, see e.g. \cite{DKV} for a definition. Thus for $\FO(\tau)$-formulae we may write $\A \models_{X_i} \phi$ instead of $\A \models_{(X_i)} \phi$.
\begin{definition}[Lax polyteam semantics]\label{def:semantics}
Let $\A$ be a $\tau$-structure and $\pt X$ a polyteam of $\A$. The satisfaction relation $\models_{\pt X}$ for poly first-order logic is defined as follows:

\begin{tabbing}
left \= $\A \models_{\pt X} (\psi \land \theta)$\, \= $\Leftrightarrow$ \= \, $\forall s\in X: s(\tuple x) \in R^{\A}$\kill
\> $\A \models_{\pt X} x= y$ \> $\Leftrightarrow$ \> if $x,y\in\var{i}$ then $\forall s\in X_i: s(x)=s(y)$\\ 
\> $\A \models_{\pt X} x \neq y$ \> $\Leftrightarrow$ \> if $x,y\in\var{i}$ then  $\forall s\in X_i: s(x) \not= s(y)$\\
\> $\A \models_{\pt X} R(\tuple x)$ \> $\Leftrightarrow$ \> if $\vec{x}\in\var{i}^k$ then  $\forall s\in X_i: s(\tuple x) \in R^{\A}$\\ 
\> $\A \models_{\pt X} \neg R(\tuple x)$ \> $\Leftrightarrow$ \>  if $\vec{x}\in\var{i}^k$ then  $\forall s\in X_i: s(\tuple x) \not\in R^{\A}$\\
\> $\A\models_{\pt X} (\psi \land \theta)$ \> $\Leftrightarrow$ \> $\A \models_{\pt X} \psi \text{ and } \A \models_{\pt X} \theta$\\
\> $\A\models_{\pt X} (\psi \lor \theta)$ \> $\Leftrightarrow$ \> $\A\models_{\pt Y} \psi \text{ and } \A \models_{\pt Z} \theta \text{ for some  $\pt{Y},\pt{Z}\sub \pt{X}$}$ s.t.\ $\pt{Y}\cup \pt{Z} = \pt{X}$\\ 
\>$\A\models_{\pt X} (\psi \lor^j \theta)$ \> $\Leftrightarrow$ \> $\A\models_{\pt{X}[Y_j/X_j]} \psi \text{ and } \A \models_{\pt{X}[Z_j/X_j]} \theta$, \\
\> \> \> for some  $Y_j,Z_j\sub X_j$ s.t. $Y_j\cup Z_j = X_j$\\
\> $\A\models_{\pt X} \forall x\psi$ \> $\Leftrightarrow$ \> $\A\models_{\pt{X}[X_i[A/x] / X_i]} \psi$, when $x\in\var{i}$ \\
\> $\A\models_{\pt X} \exists x\psi$ \> $\Leftrightarrow$ \> $\A\models_{\pt{X}[X_i[F/x] / X_i]} \psi \text{ holds for some } F\colon X_i \to \Po(A)\setminus \{\emptyset\}$,\\
\> \> \> when $x\in\var{i}$
\end{tabbing}
\end{definition}
\begin{remark}
	Note that whereas the global disjunction is both commutative and associative, the local disjunction is only commutative. In particular $(\phi \lor^i \psi) \lor^j \theta$ is not, in general, equivalent with $\phi \lor^i (\psi \lor^j \theta)$. However the local disjunction is associative with respect to local disjunctions of the same sort, i.e., $(\phi \lor^i \psi) \lor^i \theta$ and $\phi \lor^i (\psi \lor^i \theta)$ are equivalent.
\end{remark}

The truth of a \emph{sentence} $\phi$  (i.e., a formula with no free variables) in a structure  $\A$ is defined as: 
$\A \models \phi\textrm{ if }\A \models_{(\{\emptyset\})} \phi,$
where $(\{\emptyset\})$ denotes the polyteam consisting only singleton teams of the empty assignment.  We write $\fr{\phi}$ for the set of free variables in $\phi$, and $\ifr{i}{\phi}$ for $\fr{\phi}\cap\var{i}$.

Polyteam semantics is a conservative extension of team semantics in the same fashion as team semantics is a conservative extension of Tarski semantics \cite{vaananen07}.
\begin{proposition}\label{extension}
Let $\phi\in \FO(\tau)$ whose variables are all of sort  $i\in \mathbb{N}$. Let $\A$ be a $\tau$-structure and $\pt{X}$ a polyteam of $\A$. Then
\[
\A\models_{\pt X} \phi \,\Leftrightarrow\,\A\models_{X_i} \phi \,\Leftrightarrow\, \forall s\in X_i: \A\models_s \phi,
\]
where $\models_s$ denotes the ordinary satisfaction relation of first-order logic.
\end{proposition}

\begin{example}\label{ex:2}
A relational database schema
\begin{align*}
\textsc{Patient}=&\{\texttt{patient\_id,patient\_name}\},\\
\textsc{Case}=&\{\texttt{case\_id,patient\_id,diagnosis\_id,confirmation}\},\\
\textsc{Test}=&\{\texttt{diagnosis\_id,test\_id}\},\\
\textsc{Results}=&\{\texttt{patient\_id,test\_id,result}\}
\end{align*}
stores information about patient cases and their related laboratory tests. In order to maintain consistency of the stored data, database management systems support the use of integrity constraints that are based on functional and inclusion dependencies. For instance,  on  relation schema $\textsc{Patient}$ the key $\texttt{patient\_id}$ (i.e. the dependence atom $\dep{\texttt{patient\_id,patient\_name}}$) ensures that no patient id can refer to two different patient names. 
On $\textsc{Case}$ the foreign key $\texttt{patient\_id}$ referring to $\texttt{patient\_id}$ on $\textsc{Patient}$ (i.e. the inclusion atom $\textsc{Case}[\texttt{patient\_id}]\sub \textsc{Patient}[\texttt{patient\_id}]$) enforces that patient ids on $\textsc{Case}$ refer to real patients. The introduction of poly-dependence logics opens up possibilities for more expressive data constraints. 
The poly-inclusion formula
\begin{align*}
\phi_0 \dfn &\texttt{confirmation}\neq \textit{positive} \hspace{1mm}\vee^{\textsc{Case}}\exists x_1x_2\big(x_1\neq x_2 \wedge\\
&\bigwedge_{i=1,2}(\textsc{Case}[\texttt{diagnosis\_id},x_i]\sub \textsc{Test}[\texttt{diagnosis\_id,test\_id}]\wedge\\
&\hspace{0mm}\textsc{Case}[\texttt{patient\_id},x_i,positive]\sub \textsc{Results}[\texttt{patient\_id,test\_id,result}])\big)
\end{align*}
ensures that a diagnosis may be confirmed only if it has been affirmed by two different appropriate tests. The poly-exclusion formula 
\begin{align*}\phi_1 \dfn &\texttt{confirmation}\neq \textit{negative} \hspace{1mm}\vee^{\textsc{Case}}\\
&\forall x\big(\exc{\textsc{Case}[\texttt{diagnosis\_id},x]}{ \textsc{Test}[\texttt{diagnosis\_id,test\_id}]}\vee^{\textsc{Case}}\\
&\hspace{4mm}\exc{\textsc{Case}[\texttt{patient\_id},x,positive]}{ \textsc{Results}[\texttt{patient\_id,test\_id,result}]}\big)
\end{align*}
makes sure that a diagnosis may obtain a negative confirmation only if it has no positive indication by any suitable test. Note that both formulae employ local disjunction and quantified variables that refer to $\textsc{Case}$. Interestingly, the illustrated expressive gain is still computationally feasible as both $\phi_0$ and $\phi_1$ can be enforced in polynomial time.  For $\phi_0$ note that the data complexity of poly-inclusion logic is in $\PTIME$ because this logic can be translated to fixed-point logic (see Theorem \ref{thm:gfp}). For $\phi_1$ observe that satisfaction of a formula of the form $\exc{\tuple x^1}{\tuple y^2}\vee^1 \exc{\tuple x^1}{\tuple z^3}$ can be decided in $\PTIME$ as well.

\end{example}

\noindent
\textbf{Poly-dependence logics.}
\emph{Poly-dependence}, \emph{poly-independence},  \emph{poly-inclusion}, and \emph{poly-exclusion logics} ($\PFO(\pdeps)$, $\PFO(\pinds)$, $\PFO(\pincs)$, and $\PFO(\pexcs)$, resp.) are obtained by extending $\PFO$ with poly-dependence, poly-independence, poly-inclusion, and poly-exclusion atoms, respectively. In general, given a set of atoms $\mathcal{C}$ we denote by $\PFO(\mathcal{C})$  the logic obtained by extending $\PFO$ with the atoms of $\mathcal{C}$. We also consider poly-atoms in the team semantics setting; by $\FO(\mathcal{C})$ we denote the extension of first-order logic by the poly-atoms in $\mathcal{C}$. Similarly, it is also possible to consider atoms of Section \ref{sec:teams}  in the polyteam setting by requiring that the variables used with each atom are of a single sort. For two logics $\calL$ and $\calL'$ in polyteam setting, we write $\calL\leq \calL'$ if for all $\phi\in \calL$ there is $\phi'\in \calL'$ such that $\A\models_{\pt X} \phi$ if and only if $\A\models_{\pt X} \phi'$, for all structures $\A$ and polyteams $\pt X$. We also write $\calL\equiv \calL'$ if $\calL\leq \calL'$ and $\calL'\leq \calL$. We define $``$$\leq$$"$ and $``$$\equiv$$"$ analogously for two logics in the team setting.

\subsection{Basic properties}

A polyteam $\pt X$ is called \emph{strictly non-empty}, if none of the teams $X_i$, $i\in \N$, is empty.
We say that a formula $\phi$ is \emph{local} in polyteam semantics if for all $\tuple V=(V_i)_{i\in\mathbb{N}}$ where $\ifr{i}{\phi}\sub V_i$ for $i\in \N$, and all structures $\A$ and polyteams $\pt X$, we have
\[\A\models_{\pt X} \phi \Leftrightarrow \A\models_{\pt X \upharpoonright \tuple V}\phi.\]
The truth value of a local formula depends only on the free variables at each coordinate, including those coordinates where the set of free variables is empty. 
 For instance, the formula $\exists x^1 (x^1\neq x^1)$ is true if and only if the first coordinate team is empty.
 The truth value of a local formula on a strictly non-empty polyteam depends only on the values of its free variables.
 We now call a logic $\calL$ \emph{local} if all its formulae are local. 
\begin{proposition}[Locality]\label{prop:locality}
For any set $\mathcal{C}$ of  
generalised poly-atoms $\PFO(\mathcal{C})$ is local.
\end{proposition}
Furthermore, the downward closure of dependence logic as well as the union closure of inclusion logic generalise to  polyteams.
\begin{proposition}[Downward Closure and Union Closure]\label{prop:closure}
Let $\phi$ be a formula of  $\PFO(\pdeps)$, $\psi$ a formula of  $\PFO(\pincs)$, $\A$ a model, and $\pt X,\pt Y$ two polyteams. Then
$\A\models_{\pt X}\phi$  and $\pt Y \sub \pt X$ implies that   $\A\models_{\pt Y}\phi$, and $\A\models_{\pt X} \psi$ and $\A\models_{\pt Y}\psi$ implies that $\A\models_{\pt X\cup \pt Y}\psi$.
\end{proposition}

The following proposition shows that the replacement of independence (dependence) atoms with any (downwards closed) class of atoms definable in existential second-order logic ($\ESO$) results in no expressive gain, if the empty team is ignored. Note that dependence and independence
 logic formulae are always true for the empty team. We say that a formula $\phi$ over team semantics has the \emph{empty team property} if $\A\models_{\emptyset} \phi$ for all models $\A$, and a logic $\calL$ has the empty team property is all of its formulae have it.

\begin{proposition}\label{remarkprop}
Let $\mathcal{C}$ ($\mathcal{D}$, resp.) be the class of all (all downwards closed, resp.) $\ESO$-definable poly-atoms. With respect to non-empty teams, $\FO(\mathcal{C}) \equiv \FO(\inds)$, $\FO(\mathcal D) \equiv \FO(\deps)$. If the atoms in $\mathcal{C}$ and $\mathcal{D}$ have the empty team property, the restriction to non-empty teams can be lifted. Moreover $\FO(\pincs) \equiv \FO(\incs)$.
\end{proposition}
\begin{proof}
The claim $\FO(\pincs) \equiv \FO(\incs)$ follows directly from the observation that in the team semantics setting poly-inclusion atoms are exactly inclusion atoms. Note that  $\FO(\inds)$ ($\FO(\deps)$, resp.) captures all (all  downwards closed, resp.) $\ESO$-definable properties of teams which include the empty team (see Theorem \ref{thm:earlier}). It is easy to show (cf. \cite[Lemma 5]{kokuvi16}) that every property of teams definable in $\FO(\mathcal{C})$ ($\FO(\mathcal{D})$, resp.) is $\ESO$-definable ($\ESO$-definable and downwards closed, resp.).
Thus since $\inds\in\mathcal{C}$ and $\deps\in \mathcal{D}$, we obtain that $\FO(\mathcal{C}) \equiv \FO(\inds)$ and $\FO(\mathcal{D}) \equiv \FO(\deps)$ with respect to non-empty teams. Finally note that if all atoms in $\mathcal{C}$  and $\mathcal{D}$ have the empty team property, the logics  $\FO(\mathcal{C})$ and $\FO(\mathcal{D})$ have it as well. \qedd
\end{proof}

\begin{remark}\label{uniremark}
In particular it follows from the previous proposition that, in the polyteam setting, each occurrence of any (any downwards closed, resp.) $\ESO$-definable poly-atom (with the empty team property) that takes variables of a single sort as parameters may be equivalently expressed by a formula of $\PFO(\inds)$ ($\PFO(\deps)$, resp.) that only uses variables of the same single sort.
\end{remark}

We end this section by considering the relationship of global and local disjunctions. In particular, we  observe that either one of the two disjunctions could be omitted from $\PFO$ without influencing the expressivity of the logic.
To facilitate our construction, we here allow the use of disjunctions of type $\vee^I$, where $I$ is a set of indices, with the obvious semantics. We then show that $\lor$ can be replaced by $\lor^I$ and  $\lor^I$ by  $\lor^i$.

\begin{proposition}\label{onedisjunction}
For every formula $\phi$ of $\PFO$ there exists an equivalent formula $\psi_l$ ($\psi_g$, resp.) of $\PFO$ in which all disjunctions are local (global, resp.).
\end{proposition}
\begin{proof}
\textbf{Local case.}
Let $\phi$ be a formula of  $\PFO$ and let $I$ list the sorts of all the variables that occur in $\phi$. First, we let $\phi'$ denote the formula obtained from $\phi$ by substituting all occurrences of $\lor$ by $\lor^I$. It is a direct consequence of the locality property that $\phi$ and $\phi'$ are equivalent.

We will next show how to eliminate disjunctions of type $\lor^I$. Without loss of generality, we restrict to models of cardinality at least two.
Let $\phi = \psi \lor^I \theta$ be a formula of $\PFO$ and let $I=\{i_1,\dots,i_n\}$. Define
\[
\phi^+ := \bigexists_{i\in I} x^i y^i  (\xi_l \land \xi_r),
\]
 where, for each $i\in I$, the variables $x^i$ and $y^i$ are fresh and distinct, and
\begin{multline*}
\xi_l:=  (x^{i_1}= y^{i_1} \lor^{i_1} (x^{i_1} \neq y^{i_1} \land (  x^{i_2} = y^{i_2} \lor^{i_2} ( x^{i_2} \neq y^{i_2}\\
\shoveright{\land (\dots \land ( x^{i_n} = y^{i_n} \lor^{i_n} ( x^{i_n} \neq y^{i_n} \land \psi)\dots),}\\
\shoveleft{\xi_r:=  (x^{i_1} \neq y^{i_1} \lor^{i_1} ( x^{i_1} = y^{i_1} \land (x^{i_2} \neq y^{i_2} \lor^{i_2} (x^{i_2} = y^{i_2}} \\
 \land (\dots \land (x^{i_n} \neq y^{i_n} \lor^{i_n} ( x^{i_n} = y^{i_n} \land \theta )\dots).
\end{multline*}
The idea above is that the variables $x^{i_j}$, $y^{i_j}$ are used to encode a (possibly overlapping) split of the team $X_j$. Using locality it is easy to see that $\phi$  and $\phi^+$ are equivalent over structures of cardinality at least two. From this the claim follows in a straightforward manner.

\noindent
\textbf{Global case.}
We show how a single local disjunction is eliminated. From this the result follows. Again, without loss of generality, we restrict to structures of cardinality at least two. Let $\phi = \psi\lor^j \theta$ be a formula of  $\PFO$ and let $I$ list the sorts of all the variables that occur in $\phi$ except $j$. Define $\phi^*$ as
\[
\bigforall_{i\in I} x^i y^i \big((\psi \land \bigwedge_{i\in I} x^i=y^i ) \lor(\theta \land \bigwedge_{i\in I} x^i\neq y^i ) \big),
\]
where, for each $i\in I$, the variables $x^i$ and $y^i$ are fresh and distinct. The idea in $\phi^*$ is that on structures with at least two elements, the evaluation of the universal quantifiers $\forall x^i y^i$ on a polyteam $\vec{X}$ duplicates each assignment in $X_i$ in at least two ways: in some duplicates $x^i=y^i$ holds and in others $x^i\neq y^i$ holds. It then follows from locality that $\phi^*$ and $\phi$ are equivalent.
\qedd
\end{proof}

We wish to point out that, in contrast to the previous result, it is easy to see that there is no single formula in $\PFO$ without local (global, reps.) disjunction that defines global (local, resp.) disjunction.
 Moreover it's worth noticing that the translations produced above are quite involved and rely on the use of quantifiers. It is easy to define natural fragments of $\PFO$ where the disjunctions cannot be expressed with another.

\subsection{Data exchange in the polyteam setting}
As promised, we now return to the topic of modelling data exchange in our new setting. In this section we restrict our attention to poly-atoms that are embedded dependencies. Our first goal is to define the notions of \emph{source-to-target} and \emph{target} poly-atoms. For this purpose we define a normal form for embedded dependencies.
We call an embedded dependency $\forall \vec{x} \big (\phi(\vec{x}) \rightarrow \exists \vec{y} \psi(\vec{x},\vec{y})\big )$ \emph{separated} if the relation symbols that occur in $\phi$ and $\psi$ are distinct. A poly-atom is called \emph{separated}, if the defining formula is a separated embedded dependency. In the polyteam setting this is just a technical restriction as non-separated poly-atoms can be always simulated by separated ones.
Below we use the syntax $A(\vec{x}_1,\dots, \vec{x}_l, \vec{y}_1,\dots, \vec{y}_k)$ for separated poly-atoms. The idea is that $\vec{x}_i$s project extensions for relations used in the antecedent and $\vec{y}_j$s in the consequent of the defining formula.

Let $\mathcal S$ and $\mathcal T$ be a set of source relations and target relations from some data exchange instance, respectively.   Let $\tuple X=(S_1, \dots S_n, T_1,\dots,T_m)$ be a polyteam that encodes $\mathcal S$ and $\mathcal T$ in the obvious manner. We say that an instance of a separated atom $A(\vec{x}_1,\dots, \vec{x}_l, \vec{y}_1,\dots, \vec{y}_k)$ is \emph{source-to-target} if each  $\vec{x}_i$ is a tuple of variables of the sort of $S_j$, for some $j$,  and each $\vec{y}_i$ is a tuple of variables of the sort of $T_j$, for some $j$.  Analogously the instance $A(\vec{x}_1,\dots, \vec{x}_l, \vec{y}_1,\dots, \vec{y}_k)$ is \emph{target} if  each $\vec{x}_i$ and $\vec{y}_j$ is a tuple of variables of the sort of $T_p$ for some $p$.

Data exchange problems can now be directly studied in the polyteam setting. For example the \emph{existence-of-solution} problem can be reduced to a model checking problem by using first-order quantifiers to \emph{guess} a solution for the problem while the rest of the formula describes the dependences required to be fulfilled in the data exchange problem. 

\begin{example}\label{ex:3}
A relational database schemas 
\begin{align*}
&\mathcal S:\quad \textsc{P(rojects)}=\{\texttt{name, employee, employee\_position}\},\\
&\mathcal T:\quad \textsc{E(mployees)}=\{\texttt{name, project\_1, project\_2}\}
\end{align*}
are used to store information about employees positions in different projects. We wish to check whether for a given instance of the schema $\mathcal S$ there exists an instance of the schema $\mathcal T$ that does not lose any information about for which projects employees are tasked to work and that uses the attribute {\texttt{name}} as a key. The $\PFO(\pincs, \deps)$-formula
\begin{multline*}
   \phi := \exists x_1 \exists x_2 \exists x_3  
     \Big(\big( \textsc{P}[\texttt{employee, name}] \subseteq \textsc{E}[x_1,x_2] \\
     \lor^\textsc{P}
    \textsc{P}[\texttt{employee, name}] \subseteq \textsc{E}[{x_1, x_3}] \big) \land
    \dep{x_1,(x_2, x_3)} \Big),
\end{multline*}
when evaluated on a polyteam that encodes an instance of the schema $\mathcal S$, expresses that a solution for the data exchange problem exists. The variables $x_1$, $x_2$, and $x_3$ above are of sort \textsc{E} and are used to encode attribute names \texttt{name, project\_1}, and  \texttt{project\_2}, respectively. The dependence atom above enforces that the attribute {\texttt{name}} is a key.
\end{example}

\section{Relationship between  polyteam and team semantics}\label{sect:rel}
 Before embarking on an analysis of the expressive power of polyteam logics, we consider the relationship between polyteam and team semantics. One might ask is polyteam semantics really necessary in order to model the polyrelational case. Is it not possible to embed and interpret polyteams within team semantics? We will next shed some light on this question. Independence logic $\FO(\inds)$ is as expressive as existential second-order logic $\ESO$ when the team is encoded as a relation \cite{galliani12}. Moreover it is relatively clear that $\PFO(\pinds)$ translates to $\ESO$ when the polyteam is encoded as a tuple of relations. Furthermore in $\ESO$ multiple relations can be encoded into a single relation of large enough arity.  Theorem \ref{cor:polytouni}  establishes that, analogously to the case with $\ESO$, $\PFO(\pinds)$ can also be simulated using $\FO(\inds)$. On the other hand, such a result does not seem likely to hold  for  $\PFO(\pdeps)$ and $\FO(\deps)$. It seems that there is no suitable way to encode multiple relations into a single relation that is also compatible with downward closure; it is known that $\FO(\deps)$ characterises downward closed $\ESO$-definable properties of teams \cite{kontinenv09}.
Section \ref{sect:rep} lays out the groundwork for relating polyteam semantics to team semantics. Our results will be shown in Section \ref{sect:translation}.

\subsection{Team representation}\label{sect:rep}
First let us formulate precisely how polyteams can be represented by teams. 
Let $X$ be a non-empty team whose domain takes variables from $\var{1},\ldots ,\var{n}$. Then $X$ \emph{represents} the polyteam $\pt X=(X_1, \ldots ,X_n)$ where $X_i=X\upharpoonright \var{i}$.
 Note that any strictly non-empty polyteam $(X_1, \ldots ,X_n)$, where $\bigcup\Dom{X_i}$ is finite, can be represented by a team.
 In order to deal with the situations where some coordinate team $X_i$ of $\pt X$ is empty, we next introduce the concepts of \emph{contraction} and \emph{emptification}.  
\begin{definition}[Contraction and emptification]
	Let $Q$ be a generalised quantifier of type $(j_1,\dots,j_n)$. The generalised quantifier $Q'$ of type $(j_1,\dots, j_{i-1}, j_{i+1},j_n)$ is the \emph{$i$-contraction} of $Q$ if the equivalence
	\[
	 (A,R_1,\dots, R_{i-1}, \emptyset ,R_{i+1} ,R_n)\in Q \Leftrightarrow (A,R_1,\dots, R_{i-1},R_{i+1} ,R_n)\in Q'
	\]
	holds for every $A, R_1,\dots, R_n$.  For $I\subseteq \{1,\dots,n\}$,  the \emph{$I$-contraction} of $Q$ is defined analogously. We say that a class of generalised poly-atoms $\calC$ is \emph{closed under contraction} if, for any 
	poly-atom $\alpha_Q\in \calC$ and any index $i$,  $A_{Q'}\in\calC$, where $Q'$ is the $i$-contraction of $Q$.
	
	Let $\phi = \alpha_Q(\tuple x^{j_1}_1, \ldots ,\tuple x^{j_n}_n)$ be a poly-atom, $i\in\N$ a natural number, $I=\{k \mid 1\leq k \leq n, j_k=i\}$ a set of indices, 
	 and $Q'$ the \emph{$I$-contraction} of $Q$. Furthermore let $\vec{\chi}$ list the variable tuples in $\tuple x^{j_1}_1, \ldots,\tuple x^{j_n}_n$ that are not of sort $i$. We let $\emp{\phi}{i}$ denote the poly-atom $\alpha_{Q'}(\vec \chi)$, and call it the \emph{$i$-emptification} of $\phi$.
 \end{definition}
 For instance, the $1$-emptification of the dependence atom $\dep{x^1,y^1}$ is equivalent to $\top$.
 It is easy to see that in polyteam setting, the $i$-emptification of $\alpha(\vec{x})$ can be equivalently expressed with a $\PFO(\alpha)$-formula. However the same is not, in general, true in the team semantics setting.
 
 Next we show how statements on polyteams that have empty coordinate teams can be transformed to equivalent statements over strictly non-empty polyteams.

\begin{definition}[Emptification of complex formulae]
	Let $\calC$ be a set of poly-atoms, 
	$\phi$ a formula of $\PFO(\calC)$, and $i\in\N$.
	We denote by $\emp{\phi}{i}$ the formula obtained from $\phi$ by simultaneously replacing all first-order literals of sort $i$ with $\top$, all disjunctions $\lor^i$ with $\land$, and all poly-atoms $\alpha$ with $\emp{\alpha}{i}$, and by eliminating all quantifiers of sort $i$.
	We call $\emp{\phi}{i}$ the \emph{$i$-emptification of} $\phi$.
	Note that $\emp{\phi}{i}$ is free of variables and local disjuctions of sort $i$.
 \end{definition}	
  The following proposition can now be shown by straightforward structural induction.
 \begin{proposition}\label{prop:empty}
 Let $\phi\in \PFO(\calC)$, where $\calC$ is any set of poly-atoms. Then for all models $\A$ and polyteams $\pt X$, 
 \[
 \A\models_{\pt X} \emp{\phi}{i}  \iff \A\models_{\pt X_{i=\emptyset}} \phi ,
 \]
where $\pt X_{i=\emptyset}$ is obtained from $\pt X$ by substituting the empty team for $X_i$.
 \end{proposition}

Next we introduce a trick that enables us to change a representation of a polyteam on the fly.
\begin{lemma}\label{lem:er}
	Let $\calC$ be a set of poly-atoms and $X$ a  non-empty team representing $\pt X=(X_1,\dots,X_n)$.  For every $\phi\in\FO(\calC,\incs)$ there exists a formula $\phi^{\exists r}\in\FO(\calC,\incs)$ such that
	\[
	\A \models_X \phi^{\exists r} \Leftrightarrow \A \models_{Y} \phi, \text{ for some team representation $Y$ of $\pt X$.}
	\]
\end{lemma}
\begin{proof}
	Let $\phi$, $\pt{X}$, and $X$ be as described above. For each $1\leq i \leq n$, let $\tuple x^i = x^i_1,\dots,x^i_k$ be the variable domain of $X_i$ and let $\vec y^i = y^i_1,\dots,y^i_k$ be fresh and distinct variables. Let $\tuple x$ and $\tuple y$ denote  $\tuple x^1,\dots,\tuple x^n$ and $\tuple y^1,\dots,\tuple y^n$, respectively.  Define
	\[
	\phi^{\exists r} \dfn \exists \vec{y} \Big(\bigwedge_{1\leq i \leq n}  (\vec{x}^i\subseteq \vec{y}^i \land \vec{y}^i\subseteq \vec{x}^i) \land \exists \vec{x} \big( \vec{x}\subseteq \vec{y} \land \vec{y}\subseteq \vec{x} \land \phi \big) \Big).
	\]
	It is easy to check, using locality, that $\phi^{\exists r}$ defined as above satisfies the claim of the lemma.\qedd
\end{proof}

\subsection{Translation to team semantics}\label{sect:translation}

Using the concepts and results from the previous section we will now move to translations from polyteam semantics to team semantics. The next theorem reveals that   
polyteam semantics can be simulated with team semantics 
 using inclusion atoms and \emph{classical disjunction} $\cvee$:
\[\A\models_X\phi\cvee\psi \text{ iff }\A\models_X\phi\text{ or }\A\models_X\psi.\]

\begin{theorem}\label{thm:polytouni2}
 Let $\phi$ be a formula in $\PFO(\mathcal{C})$, where $\calC$ is a contraction closed set of poly-atoms. Then there is a formula $\phi^*\in \FO(\mathcal{C},\incs,\cvees)$ such that for any structure $\A$, with at least two elements, 
 and any non-empty team $X$ representing $\pt X=(X_1, \ldots ,X_n)$, 
\[
\A\models_{\pt X} \phi \Leftrightarrow  \A^*\models_{X} \phi^*,
\]
where $\A^*$ is an expansion of $\A$ with two distinct constants.
\end{theorem}
\begin{proof}
Define a rank $\rr$ of a formula $\phi$ as follows. For a poly-atom $\alpha(\tuple x_1^{i_1}, \ldots ,\tuple x_n^{i_n})$, $\rr(\alpha) \dfn n$. For first-order literals $\alpha$, $\rr(\alpha) \dfn 1$.  For complex formulae, $\rr(Q\psi) \dfn \rr(\psi)+1$ and $\rr(\psi C \theta) \dfn \rr(\psi)+\rr(\theta)+1$, where $Q\in\{\exists,\forall\}$ and $C\in \{\wedge,\vee,\vee^i\}$. We prove the claim by induction on $\rr(\phi)$ for a mapping $\phi\mapsto \phi^*$ defined recursively. This mapping is the identity on literals,
 and homomorphism on conjunction and universal quantification; proving the induction claim is straightforward for these cases and thus omitted. By Proposition \ref{onedisjunction} we can also exclude the case of global disjunction. Thus it suffices to consider local disjunction and existential quantification.

\noindent
\textbf{Local disjunction.} Suppose $\phi=\psi \vee^i \theta$. We define 
\begin{equation}\label{eq:kaava}
\phi^* \dfn ((\emp{\psi}{i})^*\wedge \theta^*) \cvee(\psi^*\wedge (\emp{\theta}{i})^*)\cvee \Phi^{\exists r},
\end{equation}
where ${}^{\exists r}$ is as defined in Lemma \ref{lem:er} and
\begin{equation}\label{eq:isophi}
\Phi \dfn \exists z \,\Big( \vec{x}0 \subseteq \vec{x} z \land \vec{x} 1 \subseteq \vec{x} z \wedge  \big((\psi^* \wedge z=0) \vee (\theta^*\wedge z=1)\big)\Big),
\end{equation}
 where $z$ is a fresh variable and
  $\tuple x$ %
   lists all the free variables of $\phi$ which are not of the sort $i$. 
  Furthermore, $0$ and $1$ are the two distinct constants, and $\tuple x0 \sub \tuple yz$ and $\tuple x0 \sub \tuple yz$ are shorthands for $\exists v (v=0 \wedge \tuple xv \sub \tuple yz)$ and $\exists v (v=1 \wedge \tuple xv \sub \tuple yz)$, respectively.

 Assume first that $\A\models_{\pt X} \phi$. Then there are $Y_i\cup Z_i=X_i$ such that $\A\models_{\pt Y} \psi$  and $\A\models_{\pt Z} \theta$, where $\tuple Y$ and $\tuple Z$ are obtained from $\tuple X$ by substituting respectively $Y_i$ and $Z_i$ for $X_i$. Suppose $Y_i$ is empty. Then by Proposition \ref{prop:empty} $\A\models_{\pt X} \emp{\psi}{i}$. Thus by induction hypothesis
  $\A\models_{X} (\emp{\psi}{i})^*$. Furthermore $\A\models_{X}\theta^*$ follows from $\A\models_{\pt Z} \theta$ by the induction hypothesis, for $\pt Z = \pt X$ since $Y_i=\emptyset$. Thus the first disjunct of $\phi^*$ holds. If $Z_i$ is empty, we similarly obtain that the second disjunct holds.  
  
 Suppose then that both $Y_i$ and $Z_i$ are non-empty.
 Let $X^*$ be a representation of $\pt X$ obtained by taking the Cartesian product of the teams $X_1,\dots,X_n$.
 For $s\in X^*$, define
 \[G(s) \dfn \begin{cases}
 \{0,1\} &\text{ if $s\upharpoonright \var{i}\in Y_i\cap Z_i$}\\
 \{0\}&\text{ if $s\upharpoonright \var{i} \in Y_i\setminus Z_i$}\\
 \{1\}&\text{ if $s\upharpoonright \var{i}\in Z_i\setminus Y_i$}.
 \end{cases}
\]
 Clearly, the two left-most conjuncts in \eqref{eq:isophi} are satisfied by $X^*[G/z]$.  Recall that  $\tuple Y$ and $\tuple Z$ were obtained from $\tuple X$ by substituting respectively $Y_i$ and $Z_i$ for $X_i$.
 Let $Y'$ and $Z'$ consist of those assignments of $X^*[G/z]$ whose restriction to $\var{i}$ belongs to $Y_i$ and $Z_i$, respectively.
 Note that if the variable $z$ is disregarded from $Y'$ and $Z'$, team representations of $\pt Y$ and $\pt Z$ are obtained, respectively.
 Now
 by induction hypothesis, locality (Proposition \ref{prop:locality}), and the selection of $G$, 
 \begin{equation}\label{eq:iff}
 \A\models_{\pt Y} \psi \,\Leftrightarrow\, \A \models_{Y'} \psi^* \,\Leftrightarrow\, \A \models_{Y'} \psi^* \land z=0.
 \end{equation}
 Similarly we obtain that $\A\models_{Z'} \theta^* \land z=1$. Since $Y'\cup Z'= X^*[G/z]$, we conclude that $\A\models_{X^*} \Phi$. Finally, by Lemma \ref{lem:er}, we obtain that $\A\models_{X} \Phi^{\exists r}$ and hence that $\A\models_{X} \phi^*$.

 For the converse direction, suppose $\A\models_{X} \phi^*$. Assume first that $\A\models_{X} (\emp{\psi}{i})^*\wedge \theta^*$. By the induction hypothesis, it follows that $\A\models_{\pt X} \theta$ and $\A\models_{\pt X} \emp{\psi}{i}$. By Proposition \ref{prop:empty} together with the semantics of local disjunctions, it follows that  $\A\models_{\pt X} \phi$. Similarly,  $\A\models_{\pt X} \phi$ if $\A\models_{X} (\emp{\theta}{i})^*\wedge \psi^*$.
Assume then that $\A\models_{X} \Phi^{\exists r}$. Then $\A\models_{Y} \Phi$ for some team $Y$ representing $\pt X$. Let $F\colon Y\to \Po(A)\setminus \{\emptyset\}$ be a mapping such that $Y[F/z]$ satisfies the quantifier-free part of \eqref{eq:isophi}, and let $U$ and $V$ consist of those assignments of $Y[F/z]$ that map $z$ to $0$ and $1$, respectively. The disjunction in \eqref{eq:isophi} now guarantees that $U\cup V= Y[F/z]$, $\A\models_U\psi^*$, and $\A\models_V\theta^*$. By the induction hypothesis, $\A\models_{\pt U} \psi$ and $\A\models_{\pt V} \theta$, where $\pt U=(U_1, \ldots ,U_n)$ and $\pt V=(V_1, \ldots ,V_n)$ are the polyteams represented by $U$ and $V$, respectively. Since $X$ and $Y$ represent the same polyteam, and $Y(\tuple x)=U(\tuple x)=V(\tuple x)$ by the inclusion atoms in \eqref{eq:isophi}, we obtain that $U_j=V_j=X_j$ for $j\neq i$. Moreover, $U_i\cup V_i= Y_i= X_i$, and thus we obtain that $\A\models_{\pt X} \psi \vee^i \theta$.

\noindent
\textbf{Existential quantification.} Suppose $\phi =\exists x \psi$ where $x\in \var{i}$. We define $\phi^* \dfn \exists x \psi^*$. Suppose first $\A\models_{\pt X}\phi$. Then $\A\models_{\pt X[X_i[F_i/x]/X_i]} \psi$ for some $F_i\colon X_i\to \Po(A)\setminus \{\emptyset\}$. If $F\colon X\to \Po(A)\setminus \{\emptyset\}$ is defined as $F(s) \dfn F_i(s\upharpoonright \var{i})$, it follows that the team $X[F/x]$ represents the polyteam $\pt X[X_i[F_i/x]/X_i]$.Hence by induction hypothesis $\A\models_{X[F/x]} \psi^*$, and thus $\A\models_X \phi^*$. 

For the converse direction, suppose $\A\models_X \phi^*$. Then we find $F\colon X\to \Po(A)\setminus \{\emptyset\}$ such that $\A\models_{X[F/x]} \psi^*$. Setting $F_i\colon X_i\to \Po(A)\setminus \{\emptyset\}$ as $F_i(s)\dfn\bigcup\{F(s')\mid s'\in X,s'\upharpoonright \var{i}=s\}$, it follows that $\pt X[X_i[F_i/x]/X_i]$ is represented by $X[F/x]$. Hence, we obtain by the induction hypothesis, that $\A\models_{\pt X[X_i[F_i/x]/X_i]} \psi$. Thus $\A\models_{\pt X}\phi$. \qedd
\end{proof}

 For poly-inclusion logic, the following corollary now follows immediately.

\begin{corollary}\label{corinc}
Let $\phi$ be a formula in $\PFO(\pincs)$. Then there is a formula $\phi^*\in \FO(\incs,\cvees)$ such that, for any structure $\A$ with at least two elements, and any non-empty team $X$ representing $\pt X$,
\[\A\models_{\pt X} \phi \Leftrightarrow  \A^*\models_{X} \phi^*,\]
where $\A^*$ is an expansion of $\A$ with two distinct constants.
\end{corollary}
\begin{proof}
By the previous theorem $\phi^*$ can be found from $\FO(\calC,\incs,\cvees)$, where $\calC$ is the closure of poly-inclusion atoms under contraction. Since poly-inclusion atoms are just inclusion atoms in team semantics, and the closure only adds poly-atoms equivalent to $\top$ or $\bot$, we obtain that $\phi^* \in \FO(\incs,\cvees)$. \qedd
\end{proof}

Next we show that the polyteam logics defined in terms of $\ESO$-definable poly-atoms can be represented in independence logic. First, we observe that $\ESO$-definable poly-atoms are closed under contraction.
\begin{proposition}\label{prop:closure}
The set of $\ESO$-definable poly-atoms is closed under contraction.
\end{proposition}
\begin{proof}
Let $\alpha_Q$ be a generalised $\ESO$-definable poly-atom of type $(j_1, \ldots ,j_n)$, and let $i\in \{1, \ldots ,n\}$. We show that the $i$-contraction of $\alpha_Q$ is definable in $\ESO$. Without loss of generality $i=n$, in which case the $i$-contraction of $\alpha_Q$, written $\alpha_{Q'}$, is of type of $(j_1, \ldots ,j_{n-1})$.
Let $\phi(R_1, \ldots ,R_{n})$ be an $\ESO$ formula which defines $\alpha_Q$, and let $\phi'(R_1, \ldots ,R_{n-1})$ be obtained from $\phi$ by replacing all relational atoms of the form $R_n(\tuple t)$ with $\bot$. Then for all models $\A=(A,R^{\A}_1, \ldots ,R^{\A}_{n-1})$,
\begin{multline*}
(A,R^{\A}_1, \ldots ,R^{\A}_{n-1})\in Q' \,\Leftrightarrow\, (A,R^{\A}_1, \ldots ,R^{\A}_{n-1},\emptyset )\in Q\\ \,\Leftrightarrow\, (A,R^{\A}_1, \ldots ,R^{\A}_{n-1},\emptyset )\models \phi
\,\Leftrightarrow\, (A,R^{\A}_1, \ldots ,R^{\A}_{n-1})\models\phi'.
\end{multline*}\qedd
\end{proof}
Second, we use the fact that  $\FO(\inds)$ characterises all $\ESO$-definable team properties.  
Note that $\rel{X}$ refers to a relation $\{s(x_1, \ldots, x_n)\mid s\in X\}$ where $x_1, \ldots ,x_n$ is some enumeration of $\Dom{X}$. 
\begin{theorem}[\cite{galliani12,kontinenv09}]\label{thm:earlier}
Let $\phi(\tuple x)$ be an $\FO(\inds)$ ($\FO(\deps)$, resp.) formula, and let $R$ be an $|\tuple x|$-ary relation.
Then there is an (downwards closed with respect to $R$, resp.) $\ESO$-sentence $\psi(R)$ such that for all teams $X\neq \emptyset$ where $\Dom{X}=\tuple x$,
\begin{equation*}
\A \models_X \phi(\tuple x) \Leftrightarrow  (\A, R:=\rel{X})\models \psi(R)
\end{equation*}
The same statement holds also vice versa.
\end{theorem}
In fact, in the above theorem we may substitute all $\ESO$-definable poly-atoms and classical disjunction for independence atoms.

\begin{theorem}\label{cor:polytouni}
Let $\calC$ be a set of $\ESO$-definable poly-atoms, and let $\phi$ be a formula in $\PFO(\mathcal{C})$. Then there is a formula $\phi^*\in \FO(\inds)$ such that for any structure $\A$ and any non-empty team $X$ representing $\pt X$,
\[
\A\models_{\pt X} \phi \Leftrightarrow  \A\models_{X} \phi^*.
\]
\end{theorem}
\begin{proof}
	We may assume that $\calC$ is the set of all $\ESO$-definable poly-atoms. Let $\phi$ be an arbitrary formula of $\PFO(\mathcal{C})$.
	Without loss of generality, it suffices to prove the above equivalence with respect to structures that have at least two elements.
	By Theorem \ref{thm:polytouni2} and Proposition \ref{prop:closure}, we find a formula $\phi' \in \FO(\mathcal{C},\incs,\cvees)$ such that, for every structure $\A$, with at least two elements, and team $X$ that represents some non-empty polyteam $\vec{X}$ of $\phi$, it holds that
	\[
	\A \models_{\pt X} \phi \Leftrightarrow \A' \models_{X} \phi',
	\]
	where $\A'$ is an expansion of $\A$ with two distinct constants $0$ and $1$. Define the formula $\phi''\dfn \exists x y \big( \dep{x} \land \dep{y} \land x\neq y \land \phi'(x/0, y/1)\big)$, where $x$ and $y$ are fresh distinct variables that do not occur in $\phi'$, and $\phi'(x/0, y/1)$ is obtained from $\phi'$ by replacing $0$ and $1$ by $x$ and $y$, respectively. Clearly
	\[
	\A' \models_{X} \phi'  \Leftrightarrow \A \models_X \phi''.
	\]
	Note that $\phi'' \in \FO(\mathcal{C},\cvees)$, for $\deps$ and $\incs$ are  $\ESO$-definable poly-atoms. Also, when restricted to non-empty teams, $\FO(\mathcal{C},\cvees)\leq \FO(\inds)$, for $\FO(\mathcal{C})\equiv \FO(\inds)$ by Proposition \ref{remarkprop}, and $\FO(\inds)$ is closed under classical disjunction because it inherits this property from $\ESO$ by Theorem \ref{thm:earlier}.
	Thus
	 we conclude that, by taking a formula $\phi^*$ of $\FO(\inds)$ which is equivalent to $\phi''$ with respect to non-empty teams, we obtain
	  the required team representation of $\phi$. \qedd 
\end{proof}

As an immediate corollary we obtain that poly-independence logic is representable in independence logic.
\begin{corollary}
Let $\phi$ be a formula in $\PFO(\pinds)$. Then there is a formula $\phi^*\in \FO(\inds)$ such that for any structure $\A$ and any non-empty team $X$ representing $\pt X$,
\[\A\models_{\pt X} \phi \Leftrightarrow  \A\models_{X} \phi^*.\]
\end{corollary}

\section{Expressiveness}\label{sec:expressivity}
The expressiveness properties of dependence, independence, inclusion, and exclusion logic and their fragments 
enjoy already comprehensive classifications. Dependence logic and exclusion logic are equi-expressive and capture all downwards closed $\ESO$ properties of teams \cite{galliani12,kontinenv09}. Independence logic, whose independence atoms violate downward closure, in turn captures all  $\ESO$ team properties \cite{galliani12}. On the other hand,  the expressivity of inclusion logic has been characterised by the so-called greatest fixed point logic \cite{gallhella13}.
In this section we turn attention to  polyteams and consider the expressivity of the poly-dependence logics introduced in this paper.
 Section \ref{sect:exp1} deals with logics with only uni-dependencies whereas in Section \ref{sect:exp2} poly-dependencies are considered.

\subsection{Uni-dependencies in polyteam semantics}\label{sect:exp1}

First we turn attention to uni-atoms in polyteam semantics. We show that with uni-atoms no interaction between different relations is possible.

\begin{theorem}\label{thm:conj}
Let $\calC$ be a set of  uni-atoms. Each formula $\phi(\tuple x^1, \ldots ,\tuple x^n)\in \PFO(\mathcal{C})$ can be associated with a sequence of formulae $\psi_1(\tuple x^1), \ldots ,\psi_n(\tuple x^n) \in \FO(\mathcal{C})$ such that for all structures $\A$ and all $\pt X=(X_1, \ldots ,X_n)$, where $X_i$ is a team with domain $\tuple x^i$,
\[\A\models_{\pt X} \phi(\tuple x^1, \ldots ,\tuple x^n) \Leftrightarrow \forall i=1, \ldots ,n: \A\models_{X_i} \psi_i(\tuple x^i).\]
Similarly, the statement holds vice versa.
\end{theorem}
\begin{proof}
The latter statement is clear as it suffices to set $\phi(\tuple x^1, \ldots ,\tuple x^n):= \psi_1(\tuple x^1)\wedge \ldots \wedge \psi_n(\tuple x^n)$. For the other direction, we define recursively functions $f_i$ that map formulae $\phi(\tuple x^1, \ldots ,\tuple x^n) \in \PFO(\mathcal{C})$ to formulae $\psi_i(\tuple x^i) \in \FO(\mathcal{C})$. By Proposition \ref{onedisjunction} we may assume that only disjunctions of type $\lor^i$, for some $i\in\mathbb{N}$, may occur in $\phi$. The functions $f_i$ are defined as follows:
\begin{itemize}
\item If $\phi(\tuple x^j)$ is an atom, then $f_i(\phi) \dfn \begin{cases}\phi &\textrm{if }i=j,\\\top&\textrm{otherwise.}\end{cases}$
\item $f_i(\psi \lor^j \theta) \dfn \begin{cases}f_i(\psi) \lor f_i(\theta) &\textrm{if }i=j,\\ f_i(\psi) \land f_i(\theta) &\textrm{otherwise.}\end{cases}$
\item $f_i(\psi \land \theta) \dfn f_i(\psi)\land f_i(\theta)$.
\item For $Q\in \{\exists,\forall\}$, set $f_i(Q x^j \psi)\dfn\begin{cases}Q x f_i(\psi) &\textrm{if }i=j,\\ f_i(\psi)  &\textrm{otherwise.}\end{cases}$
\end{itemize}

We set $\psi_i:=f_i(\phi)$ and show the claim by induction on the structure of the formula. The cases for atoms and conjunctions are trivial. We show the case for $\lor^i$.

Let $\phi=\psi\lor^j\theta$ and assume that the claim holds for $\psi$ and $\theta$. Now
\begin{align*}
\A \models_{\pt X} \phi \quad\text{ iff }\quad& \text{$\A \models_{\pt X[Y_j/X_j]}\psi$ and  $\A \models_{\pt X[Z_j/X_j]}\theta$,} \\
&\text{for some $Y_j,Z_j\subseteq X_j$ such that $Y_j\cup Z_j = X_j$}.
\end{align*}
By the induction hypothesis,  $\A \models_{\pt X[Y_j/X_j]}\psi$ and  $\A \models_{\pt X[Z_j/X_j]}\theta$ iff $\A \models_{Y_j} f_j(\psi)$, $\A \models_{Z_j} f_j(\theta)$, and
$\A \models_{X_i} f_i(\psi), \A \models_{X_i} f_i(\theta)$  for each $i\neq j$. Thus we obtain that $\A \models_{\pt X} \phi$ holds iff
\[
\A \models_{X_j} f_j(\psi)\lor f_j(\theta),  \text{ and } \A \models_{X_i} f_i(\psi) \text{ and } \A \models_{X_i} f_i(\theta)  \text{ for each } i\neq j.
\]
The above can be rewritten as
\[
\A \models_{X_j} f_j(\psi)\lor f_j(\theta),  \text{ and } \A \models_{X_i} f_i(\psi) \land f_i(\theta)  \text{ for each } i\neq j.
\]
The claim now follows, since  $f_j(\psi)\lor f_j(\theta)= f_j(\psi\lor^j \theta)$ and  $f_i(\psi)\land f_i(\theta)= f_i(\psi\lor^j \theta)$, for $i\neq j$. 

The cases for the quantifiers are similar. \qedd
\end{proof}

Theorem \ref{thm:conj} implies that poly-atoms which describe relations between two teams are beyond the scope of uni-logics. The following proposition illustrates this for $\PFO(\deps)$.

\begin{proposition} The poly-constancy atom $\pcon{x^1}{x^2}$  cannot be expressed in  $\PFO(\deps)$.
 \end{proposition}
\begin{proof}
Assume that  $\pcon{x^1}{x^2}$ can be defined by some $\phi(x^1,x^2)\in \PFO(\deps)$. By Theorem \ref{thm:conj}
there are $\FO(\deps)$-formulae $\psi_1( x^1)$ and $\psi_2( x^2)$ such that for all $\pt X=(X_1,X_2)$, where $X_i$ is a team with domain $ x^i$, it holds that 
\begin{equation}\label{polyconst}
\A\models_{\pt X} \pcon{x^1}{x^2} \Leftrightarrow \forall i=1,2 : \A\models_{X_i} \psi_i( x^i).
\end{equation}
Define teams $X_1:=\{x^1\mapsto 0\}$, $X_2:=\{x^2\mapsto 0\}$, $Y_1:=\{x^1\mapsto 1\}$, and $Y_2:=\{x^2\mapsto 1\}$. 
Now clearly $\A\models_{(X_1,X_2)} \pcon{x^1}{x^2}$, and $\A\models_{(Y_1,Y_2)} \pcon{x^1}{x^2}$.  Hence by \eqref{polyconst}, we obtain first that
\(    \A\models_{X_1} \psi_i( x^1) \textrm{ and  } \A\models_{Y_2} \psi_i( x^2) \), 
and then that  $\A\models_{(X_1,Y_2)} \pcon{x^1}{x^2}$, which is a contradiction.\qedd
\end{proof}

It is now easy to see that  Theorems \ref{thm:conj} and \ref{thm:earlier} together imply that  $\PFO(\inds)$ ($\PFO(\deps)$, resp.) captures all conjunctions of (downward closed, resp.) $\ESO$ properties of teams.
\begin{theorem}\label{thm:udepA}
Let $\phi(\tuple x^1, \ldots ,\tuple x^n)$ be a $\PFO(\inds)$ ($\PFO(\deps)$, resp.) formula where  $\vec{x}^i$ is a sequence of variables from $\var{i}$. Let $R_i$ be an $|\tuple x^i|$-ary relation symbol for $i=1, \ldots ,n$.
Then there are (downwards closed with respect to $R_i$, resp.) $\ESO$-sentences $\psi_1(R_1), \ldots ,\psi_n(R_n)$ such that for all polyteams $\pt X=(X_1, \ldots ,X_n)$ where $\Dom{X_i}=\tuple x^i$ and $X_i\neq \emptyset$
\begin{multline*}
\A \models_{\pt X} \phi(\tuple x^1, \ldots ,\tuple x^n) 
\Leftrightarrow  (\A, R_1:=\rel{X_1}, \ldots ,R_n:=\rel{X_n})\models \psi_1(R_1)\wedge \ldots \wedge \psi_n(R_n).
\end{multline*}
The same statement holds also vice versa.
\end{theorem}

\subsection{Poly-dependencies in polyteam semantics}\label{sect:exp2}
Next we consider poly-dependencies in polyteam semantics. We begin by observing that many translations between different team logics carry over to polyteam logics.

\begin{lemma}
The following equivalences hold:
\begin{align}
\pdep{\tuple x^1}{\tuple y^1}{\tuple u^2}{\tuple v^2}\equiv& \hspace{2mm}\indep{\tuple x^1, \tuple u^2 / \tuple x^1}{\tuple y^1 / \tuple y^1}{\tuple v^2 / \tuple y^1},\label{e1}\\
\pdep{\tuple x^1}{y^1}{\tuple u^2}{v^2}\equiv& \hspace{2mm}\forall z^1(y^1 = z^1 \vee^1 \exc{\tuple x^1 z^1}{\tuple u^2 v^2}),\label{e2}\\
\tuple x^1\sub \tuple u^2 \equiv& \hspace{2mm} \indepc{\tuple x^1 / \tuple u^2}{\emptyset / \emptyset},\label{e3}\\
\tuple x^1\sub \tuple u^2 \equiv& \hspace{2mm} \forall \tuple v^2(\exc{\tuple x^1}{\tuple v^2} \vee^2 \tuple v^2\sub \tuple u^2),\label{e4}\\
\exc{\tuple x^1}{\tuple u^2} \equiv&\hspace{2mm} \exists y^1 z^1 v^2 w^2 (\pdep{\tuple x^1}{y^1z^1}{\tuple u^2}{v^2w^2} \label{e5}\\
& \quad \wedge y^1=z^1 \wedge v^2\neq w^2),\notag\\
\exc{\tuple x^1}{\tuple u^2} \equiv& \hspace{2mm}  \exists \tuple y^1(\tuple u^2 \sub \tuple y^1 \wedge \exc{\tuple x^1}{\tuple y^1}),\label{e6}
\end{align}
\begin{align}
\indep{\tuple x^2,\tuple x^3 / \tuple x^1}{\tuple y^2 / \tuple y^1}{\tuple z^3 / \tuple z^1}\equiv
&\hspace{2mm}\forall \tuple p^2\tuple q^2 \exists u^2 v^2\forall \tuple p^3\tuple q^3\tuple r^3 \exists u^3 v^3 \Big( \label{e8}\\
&\pdep{\tuple p^2\tuple q^2}{u^2v^2}{\tuple p^3\tuple q^3 }{u^3v^3} \notag\\
& \wedge \big(u^2 =v^2\vee^1(u^2\neq v^2 \wedge \exc{\tuple x^2 \tuple y^2}{\tuple p^2\tuple q^2})\big) \nonumber\\
& \wedge \big(u^3 \neq v^3 \vee^2 \exc{\tuple x^3 \tuple z^3}{\tuple p^3\tuple r^3} \vee^2 \tuple p^3\tuple q^3\tuple r^3 \sub \tuple x^1\tuple y^1\tuple z^1\big)\Big).\nonumber
\end{align}
\end{lemma}
\begin{proof}
The equivalences \eqref{e1}--\eqref{e6} are straightforward and \eqref{e8}  is analogous to the corresponding  translation in the team semantics setting (see \cite{galliani12}).\qedd
\end{proof}

The following theorem compares the expressive powers of different polyteam logics. Observe that the expressivity of the logics with two poly-dependency atoms remains the same even if either one of the atoms has the standard team semantics interpretation. 
\begin{theorem}\label{cor:rel} The following equivalences of logic hold:
\begin{enumerate}[(1)]
\item $\PFO(\pdeps) \equiv \PFO(\pexcs)$,
\item $\PFO(\pinds) \equiv \PFO(\pexcs,\incs) \equiv \PFO(\pincs,\excs) \equiv \PFO(\pdeps,\incs)$\\ $\equiv \PFO(\pincs,\deps) \equiv \PFO(\pdeps,\inds) \equiv \PFO(\pexcs,\inds) \equiv \PFO(\pincs,\inds)$.
\end{enumerate}
\end{theorem}
\begin{proof}
Item (1) follows by equivalences \eqref{e2} and \eqref{e5}. Item (2) follows from the following list of relationships:
\begin{itemize}
\item $\PFO(\pinds)\sub \PFO(\pexcs,\incs)$ by \eqref{e2}, \eqref{e4}, and \eqref{e8}.
\item $\PFO(\pexcs,\incs)  \equiv \PFO(\pincs,\excs)$ by \eqref{e4} and \eqref{e6}.
\item $\PFO(\pexcs,\incs) \equiv  \PFO(\pdeps,\incs)$ by \eqref{e2} and \eqref{e5}.
\item $\PFO(\pincs,\excs) \equiv \PFO(\pincs,\deps)$, since exclusion (dependence, resp.) atoms can be described in $\FO(\deps)$ ($\FO(\excs)$, resp.) \cite{galliani12}.
\item $\PFO(\pdeps,\incs)\sub \PFO(\pdeps,\inds)$, $\PFO(\pexcs,\incs)\sub \PFO(\pexcs,\inds)$, and $\PFO(\pincs,\deps)\sub \PFO(\pincs,\inds)$ since inclusion atoms can be described in $\FO(\inds)$ \cite{galliani12} and dependence atoms by independence atoms \cite{gradel13}.
\item $\PFO(\pdeps,\inds)\sub \PFO(\inds)$, $\PFO(\pexcs,\inds)\sub \PFO(\inds)$, and $\PFO(\pincs,\inds)$\linebreak$\sub \PFO(\pinds)$ by \eqref{e1}, \eqref{e3}, and \eqref{e5}.
\end{itemize}\qedd
\end{proof}

Next we show the analogue of Theorem \ref{thm:earlier} for polyteams.

\begin{theorem}\label{thm:pindA}
 Let $\phi(R_1, \ldots ,R_n)$ be an $\ESO$-sentence. There is a $\PFO(\pdeps,\incs)$-formula $\phi^*(\tuple x^1, \ldots ,\tuple x^n)$, where $|\tuple x^i|=\ar{R_i}$, such that for all structures $\A$ and all polyteams $\pt X=(X_1, \ldots ,X_n)$ with $\Dom{X_i}=\tuple x^i$ and $X_i\neq \emptyset$,
\begin{equation*}
\A \models_{\pt X} \phi^*(\tuple x^1, \ldots ,\tuple x^n) \Leftrightarrow  (\A, R_1:=\rel{X_1}, \ldots ,R_n:=\rel{X_n})\models \phi(R_1, \ldots ,R_n).
\end{equation*}
The statement holds also vice versa.
\end{theorem}
\begin{proof}
Considering first the direction from $\PFO(\pdeps,\incs)$ to $\ESO$, let $\phi(\tuple x^1, \ldots ,\tuple x^n) $ be a $ \PFO(\pdeps,\incs)$-formula. Since poly-dependence and uni-inclusion atoms are $\ESO$-definable, $\phi$ can be represented by some $\FO(\inds)$-formula $\phi^*$ (Theorem \ref{cor:polytouni}), which in turn can be expressed as some $\ESO$-formula $\psi(R)$ (Theorem \ref{thm:earlier}). Let $R_1, \ldots ,R_n$ be fresh relation symbols with respective arities $|\tuple x^1|, \ldots ,|\tuple x^n|$. Let $\psi'$ be obtained from $\psi$ by replacing each atom $R(\tuple v_1,\ldots ,\tuple v_n)$, where $\tuple v_1,\ldots ,\tuple v_n$ are tuples of variables respective lengths $|\tuple x^1|, \ldots ,|\tuple x^n|$,
 with the conjunction $R_1(\tuple v_1)\wedge \ldots \wedge R_n(\tuple v_n)$. Then we observe that for all models $\A$ and  polyteams $\pt X=(X_1, \ldots ,X_n)$ represented by $X$,
\[\A\models_{\pt X} \phi \iff \A\models_X \phi^*\iff (\A,\rel{X})\models \psi \iff (\A,\rel{X_1},\ldots ,\rel{X_n})\models \psi'.\]

 Consider then the opposite direction. Analogously to \cite{galliani12}, we can rewrite $\phi(R_1, \ldots ,R_n)$ as
\[\exists \tuple f \forall \tuple u \big( \bigwedge_{i=1}^{n} (R_i(\tuple u_i) \leftrightarrow f_{2i-1}(\tuple u_i)=f_{2i}(\tuple u_i))\wedge \psi(\tuple u, \tuple f)\big)\]
where $\tuple f=f_1, \ldots ,f_{2n}, \ldots ,f_m$ is a list of function variables, $\psi$ is a quantifier-free formula in which no $R_i$ appears, each $\tuple u_i$ is a subsequence of $\tuple u$, and each $f_i$ occurs only as $f_i(\tuple u_{j_i})$ for some fixed tuple $\tuple u_{j_i}$ of variables. For instance, $j_i= i/2$ for even $i\leq 2n$.

Let  $\tuple b^i$ be sequences of variables of sort $i$ such that $ |\tuple b^i|=|\tuple u_i| $, and let  $\tuple u^1 \tuple y^1$ be a sequence of  variables of sort $1$ such that  $\tuple u^1$ is a  copy of $\tuple u$ and $\tuple y^1=y^1_1, \ldots, y^1_m$.  We define $\phi^*(\tuple x^1, \ldots ,\tuple x^n)$ as the formula
\[  \forall \tuple b^1\exists z^1_0z^1_1\ldots \forall \tuple b^n\exists z^n_0z^n_1 \forall \tuple u^1\exists \tuple y^1\big( \theta_0 \wedge \theta_1\wedge \psi'(\tuple u^1, \tuple y^1) )
\]
where
\begin{align*}
\theta_0 := &\bigwedge_{i=1}^n \dep{\tuple b^i,z^i_0} \wedge \dep{\tuple b^i,z^i_1} \wedge ((\tuple b^i \sub \tuple x^i\wedge z^i_0 = z^i_1) \lor^i  (\tuple x^i \mid \tuple b^i \wedge z^i_0\neq z^i_1)),\\
\theta_1:=&  \bigwedge_{i=1}^n\pdep{\tuple u^1_i}{y^1_{2i-1}}{\tuple b^i}{z^i_0} \wedge \pdep{\tuple u^1_i}{y^1_{2i}}{\tuple b^i}{z^i_1} \wedge \bigwedge_{i=n+1}^m \dep{\tuple u^1_{j_i},y^1_i},
\end{align*}
and $\psi'(\tuple u^1, \tuple y^1)$ is obtained from $\psi(\tuple u, \tuple f)$ by replacing $\tuple u$ pointwise with $\tuple u^1$ and each $f_i(\tuple u_{j_i})$ with $y^1_i$.  Above, $\theta_0$ amounts to the description of the characteristic functions $f_{2i-1}$ and $f_{2i}$. 
We refer the reader to \cite{galliani12} to check that $\A \models_{\pt X} \theta_0$ iff for all $i$ the functions $s(\tuple b^i)\mapsto s(z^i_0)$ and $ s(\tuple b^i)\mapsto s(z^i_1)$ determined by the assignments $s\in X_i$ agree on $s(\tuple b^i)$ exactly when $s(\tuple b^i)\in \rel{X_i}$. The poly-dependence atoms in $\theta_1$ then transfer these functions over to the first team, and the dependence atoms in $\psi_1$ describe the remaining functions. 
As in \cite{galliani12}, it can now be seen that $\phi^*$ correctly simulates $\phi$. Since exclusion atoms can be expressed in dependence logic, the claim then follows.\qedd
\end{proof}
By item (2) of Theorem \ref{cor:rel} the result of Theorem \ref{thm:pindA}   extends to a number of other logics as well. For instance, we obtain that poly-independence logic captures all $\ESO$ properties of polyteams.
The proof of Theorem \ref{thm:pindA} can be now  easily adapted to show that poly-exclusion and poly-dependence logic  capture all downwards closed $\ESO$ properties of polyteams.

\begin{theorem}\label{thm:pdepA}
 Let $\phi(R_1, \ldots ,R_n)$ be an $\ESO$-sentence that is downwards closed with respect to $R_i$. Then there is a $\PFO(\pdeps)$-formula $\phi^*(\tuple x^1, \ldots ,\tuple x^n)$, where $|\tuple x^i|=\ar{R_i}$, such that for all structures $\A$ and all polyteams $\pt X=(X_1, \ldots ,X_n)$ with $\Dom{X_i}=\tuple x^i$ and $X_i\neq \emptyset$,
\begin{equation*}
\A \models_{\pt X} \phi^*(\tuple x^1, \ldots ,\tuple x^n) \Leftrightarrow  (\A, R_1:=\rel{X_1}, \ldots ,R_n:=\rel{X_n})\models \phi(R_1, \ldots ,R_n).
\end{equation*}
The statement holds also vice versa.
\end{theorem}
\begin{proof}
Consider first the direction from $\PFO(\pdeps)$ to $\ESO$. By Theorem \ref{cor:rel}, $\PFO(\pdeps)$ is subsumed by $\PFO(\pinds)$, and thus the previous theorem yields a suitable $\ESO$-sentence $\phi(R_1, \ldots ,R_n)$. Since $\PFO(\pdeps)$ is downwards closed (Proposition \ref{prop:closure}), this sentence is also downwards closed with respect to $R_i$.

For the other direction, let  $\phi(R_1, \ldots ,R_n)$ be an $\ESO$-sentence in which the relations $R_i$ appear only negatively. As in the proof of Theorem \ref{thm:pindA} and by downward closure we may transform it to an equivalent form (see \cite{kontinenv09} for details)
 \[\exists \tuple f \forall \tuple u \big( \bigwedge_{i=1}^{n} (\neg R_i(\tuple u_i) \vee f_{2i-1}(\tuple u_i)=f_{2i}(\tuple u_i))\wedge \psi(\tuple u, \tuple f)\big)\]
Now the  translation $\phi(\tuple x^1, \ldots ,\tuple x^n)$ is defined analogously to the proof of Theorem \ref{thm:pindA} except for $\theta_0$ which is redefined as
\[ \theta_0 := \bigwedge_{i=1}^n \dep{\tuple b^i,z^i_0} \wedge \dep{\tuple b^i,z^i_1} \wedge (\tuple x^i \mid \tuple b^i \lor^i  z^i_0= z^i_1).\]
Finally the claim follows by eliminating the exclusion atoms  from $\theta_0$.  \qedd
\end{proof}

Next we turn to poly-inclusion logic. Over sentences, inclusion logic is known to be as expressive as positive greatest fixed point logic ($\posgfp$), the fragment of  greatest fixed point logic in which all fixed point operators occur within the scope of an even number of negations. Moreover, all team properties definable in inclusion logic are also definable in $\posgfp$, but the converse does not hold due to union closure.

\begin{theorem}[\cite{gallhella13}]\label{galhella}
Every  $\FO(\incs)$-sentence is equivalent to some $\posgfp$-sentence, and vice versa. Moreover, for every $\FO(\incs)$-formula  $\phi(\tuple x)$ there is a $\posgfp$-sentence $\psi(R)$, where $|\tuple x|=\ar{R}$,  such that  for all structures $\A$ and all teams $X$ with $\Dom{X}=\tuple x$,
\[\A\models_X \phi(\tuple x)  \iff (\A, R:=\rel{X})\models \psi(R).\]
\end{theorem}
To generalise these results to polyteam semantics, we use the fact that
 all team connectives and quantifiers distribute over classical disjunctions.
\begin{proposition}[\cite{Galliani16}]\label{prop:neg}
Let $\phi$ be an $\FO(\calC,\cvees)$-formula, where $\calC$ is any set of atoms. Then $\phi$ is equivalent to some formula of the form $\psi_1\cvee \ldots \cvee\psi_n$ where $\psi_1,\ldots ,\psi_n$ are $\FO(\calC)$-formulae.
\end{proposition}

\begin{theorem}\label{thm:gfp}
Every  $\PFO(\pincs)$-sentence is equivalent to some $\posgfp$-sentence, and vice versa. Moreover, for every $\PFO(\pincs)$-formula  $\phi(\tuple x^1, \ldots ,\tuple x^n)$ there is a $\posgfp$-sentence $\psi(R_1, \ldots ,R_n)$, where $|\tuple x^i|=\ar{R_i}$,  such that  for all structures $\A$ and all polyteams $\pt X=(X_1, \ldots ,X_n)$ with $\Dom{X_i}=\tuple x^i$,
\[\A\models_{\pt X} \phi(\tuple x^1, \ldots ,\tuple x^n)  \iff (\A, R_1:=\rel{X_1}, \ldots ,R_n:=\rel{X_n})\models \psi(R_1, \ldots ,R_n).\]
\end{theorem}

\begin{proof}
Let $\phi(\tuple x^1, \ldots ,\tuple x^n)\in \PFO(\pincs)$ be a formula, and let $\phi^*(\tuple x^1, \ldots ,\tuple x^n) \in \FO(\incs,\cvees)$ be its team representation, obtained by  Corollary \ref{corinc}, in which two additional constants $0$ and $1$ occur. Without loss of generality, we may restrict our attention to structures with at least two elements. By Proposition \ref{prop:neg} $\phi^*$ is equivalent to a disjunction $\psi_1\cvee \ldots \cvee\psi_n$, where $\psi_i$ are $\FO(\incs)$-formulae.
By Theorem \ref{galhella} each $\psi_i$ is  equivalent to some $\posgfp$-sentence $\Phi_i(R)$. Define 
\[
\Phi(R):=\exists yz\Big(y\neq z\wedge \big(\Phi_1(y/0, z/1)\vee\ldots \vee \Phi_n(y/0, z/1)\big)\Big),
\]
 where $\Phi_i(y/0, z/1)$ are obtained from $\Phi_i$ by substituting $y$ and $z$ respectively for $0$ and $1$. 
 Let $R_1, \ldots ,R_n$ be fresh relation symbols with respective arities $|\tuple x^1|, \ldots ,|\tuple x^n|$.  Let $\Phi'$ be the formula obtained from $\Phi$ by replacing each atom $R(\tuple y_1,\ldots ,\tuple y_n)$, where $\tuple y_1,\ldots ,\tuple y_n$ are tuples of variables with respective lengths $|\tuple x^1|, \ldots ,|\tuple x^n|$, with the conjunction $R_1(\tuple y_1)\wedge \ldots \wedge R_n(\tuple y_n)$.
 The following equivalence holds for all structures $\A$, with at least two elements, and strictly non-empty polyteams $\pt X=(X_1, \ldots ,X_n)$. Let $X$ denote the team representation of $\pt X$ obtained by taking the Cartesian product of the teams $X_i$, $1\leq i \leq n$. 
\begin{align*}
\A\models_{\pt X} \phi& \,\Leftrightarrow\, \A^*\models_X  \psi_1\cvee \ldots \cvee \psi_n
\\
&\,\Leftrightarrow\, (\A,R:=\rel{X})\models \Phi 
\\ 
& \,\Leftrightarrow\, \big(\A,R_1 \dfn \rel{X_1},\ldots , R_n \dfn \rel{X_n}\big)\models \Phi',
\end{align*}
where $\A^*$ is an expansion of $\A$ with two distinct constants $0$ and $1$. The converse direction for $\posgfp$-sentences follows by Theorem \ref{galhella} and since $\FO(\incs)$ is a fragment of $\PFO(\pincs)$.\qedd
\end{proof}

\section{Conclusion}
In this article we have laid the foundations of polyteam semantics in order to facilitate the fruitful exchange of ideas and results between team semantics and database theory. Our results show that many of the familiar properties and results from team semantics carry over to the polyteam setting. In particular, we identified a natural polyteam analogue of dependence atoms and gave a complete axiomatisation for the associated implication problem. We also showed that polyteam semantics can sometimes be reduced to team semantics, although it can be questioned whether such an interpretation is reasonable in the first place. The examples of this paper demonstrate that polyteam semantics is a conceptually more natural framework for capturing properties of sets of relations. Specifications for multiple relations are easier to parse if different relations are explicitly distinguished in formulae. Also, if polyteam logics are interpreted as data constraint languages, as in Example \ref{ex:2}, then the reduction from polyteam semantics to team semantics incurs an unnecessary computational overhead.
 First, a single team to represent the polyteam has to be constructed, e.g., by taking a Cartesian product of all coordinate teams; and second, this team has to be validated against a team logic formula which is much larger in size than the initial polyteam formula.

Our results also open up interesting avenues for further research. One question is to determine whether poly-dependence logic reduces to dependence logic. Our methods only work for poly-independence and poly-inclusion logic, and the proviso in the latter case was to include classical disjunction. As inclusion logic with classical disjunction is not union closed, it would also be interesting to study the team properties definable in this logic. Apart from poly-dependence atoms, we did not consider axioms for any other poly-atoms. Since the axioms of poly-inclusion atoms are already known from database theory, a natural next step  would be to axiomatise marginal poly-independence atoms. It would also be interesting to develop axiomatic methods for more expressive fragments of polyteam logics (cf. \cite{DBLP:journals/apal/Hannula15,DBLP:journals/apal/KontinenV13,Yang19}).


\bibliography{biblio,multisets}
\bibliographystyle{plain}

\end{document}